\documentclass[11  pt, oneside]{article}   	\usepackage[margin=1.25in,lmargin=0.75in,rmargin=1.75in,bmargin=0.75in,tmargin=0.75in]{geometry}      
\usepackage{graphicx}
\usepackage{booktabs} 
\usepackage{listings}
\usepackage{amsthm}
\newtheorem{theorem}{Theorem}

\newtheorem{lemma}[theorem]{Lemma}

\usepackage{algorithm}
\usepackage{algorithmic}

\newcommand{\alg}[1]{\textsc{#1}}
\newcommand{\cut}{\textbf{cut}}
\newcommand{\vol}{\textbf{vol}}

\renewcommand{\vec}[1]{\boldsymbol{\mathrm{#1}}}

\newcommand{\caln}{\mathcal{N}}

\providecommand{\vp}{\ensuremath{\vec{p}}}

\newcommand{\barS}{\bar{S}}
\newcommand{\barR}{\bar{R}}

\usepackage{amsmath}
\usepackage{amssymb}
\usepackage{paralist}
\usepackage{graphicx}
\usepackage{subfig}
\usepackage{booktabs} 

\usepackage{microtype}

\usepackage[T1]{fontenc}
\usepackage[utf8]{inputenc}
\usepackage{url}

\begin{document}
\title{Flow-Based Local Graph Clustering \\with Better Seed Set Inclusion}

\author{Nate Veldt \thanks{Mathematics Department, Purdue University. Email: \emph{lveldt@purdue.edu}. Supported by NSF award CCF-1149756. Part of the research was conducted while the author was a student computation intern at Lawrence Livermore National Laboratory.}\\
	\and
	Christine Klymko \thanks{Center for Applied Scientific Computing, Lawrence Livermore National Laboratory. Email: \emph{klymko1@llnl.gov}. This work was performed under the auspices of the U.S. Department of Energy by Lawrence Livermore National Laboratory under Contract DE-AC52-07NA27344.} \\
	\and
	David F. Gleich \thanks{Dept.\ of Computer Science, Purdue University. Email: \emph{dgleich@purdue.edu}. Supported by DARPA SIMPLEX and NSF award CCF-1149756, IIS-1422918, IIS-1546488, CCF093937
		and the Sloan Foundation. }}
\date{}


\date{}

\maketitle

\begin{abstract}
	Flow-based methods for local graph clustering have received significant recent attention for their theoretical cut improvement and runtime guarantees. In this work we present two improvements for using flow-based methods in real-world semi-supervised clustering problems. Our first contribution is a generalized objective function that allows practitioners to place strict and soft penalties on excluding specific seed nodes from the output set. This feature allows us to avoid the tendency, often exhibited by previous flow-based methods, to contract a large seed set into a small set of nodes that does not contain all or even most of the seed nodes. Our second contribution is a fast algorithm for minimizing our generalized objective function, based on a variant of the push-relabel algorithm for computing preflows. We make our approach very fast in practice by implementing a global relabeling heuristic and employing a warm-start procedure to quickly solve related cut problems. In practice our algorithm is faster than previous related flow-based methods, and is also more robust in detecting ground truth target regions in a graph, thanks to its ability to better incorporate semi-supervised information about target clusters.
\end{abstract}

\section{Introduction}
Local graph clustering is the task of finding tightly connected clusters of nodes nearby a set of seed vertices in a large graph. This task has been applied to solve problems in information retrieval~\cite{LangRao2004}, image segmentation~\cite{Mahoney:2012:LSM:2503308.2503318,VeldtGleichMahoney2016}, and community detection~\cite{AndersenLang2006,KloumannKleinberg2014}, among many other applications. In practice, seed nodes represent semi-supervised information about a hidden target cluster, and the goal is to recover or detect this cluster by combining knowledge of the seed set with observations about the topological structure of the network.

One popular approach for graph clustering is to apply flow-based algorithms, which repeatedly solve regionally biased minimum cut and maximum flow problems on the input graph. These methods satisfy very good theoretical cut improvement guarantees with respect to quotient-style clustering objectives such as conductance~\cite{AndersenLang2008,LangRao2004,OrecchiaZhu2014,VeldtGleichMahoney2016}. Additionally, some of these methods are strongly local, i.e.\ their runtime depends only on the size of the seed set and not the entire input graph.

Despite these attractive theoretical properties, existing flow-based methods exhibit drawbacks when it comes to solving real-world graph clustering problems. For example, in some cases these methods exhibit the tendency to discard important semi-supervised information in favor of optimizing a quotient-style clustering objective. More specifically, existing methods either shrink a seed set of nodes into a subset with better cut-to-size ratio~\cite{LangRao2004}, or try to find a good output cluster that overlaps well with the seed set but may not include all or even a majority of the seed nodes~\cite{AndersenLang2008,OrecchiaZhu2014,VeldtGleichMahoney2016}. While this is beneficial for obtaining theoretically good graph cuts, it is not always desirable in label propagation and community detection applications where the goal is to grow a set of seed nodes into a larger community. In addition to this, the previously cited flow-based methods treat all seed nodes equally, whereas in practice there may be varying levels of confidence for whether or not a seed node is a true representative of the undetected target cluster.

Although recently developed strongly-local methods constitute a major advancement in flow-based clustering, these also exhibit drawbacks in terms of implementation and practical performance. The \alg{LocalImprove} algorithm of Orecchia and Zhou~\cite{OrecchiaZhu2014} is known to have an extremely good theoretical runtime but relies on a complicated variation of Dinic's max-flow algorithm~\cite{Dinitz-1970-max-flow} that is difficult to implement in practice. In more recent work we developed an algorithm called \alg{SimpleLocal}, which provides a simplified framework for optimizing the same objective~\cite{VeldtGleichMahoney2016}. While this method is easy to implement and reasonably fast in practice, it still relies on repeatedly solving numerous exact maximum flow problems, and takes no advantage of warm-start solutions between consecutive flow problems that are closely related.

\paragraph{Our Contributions}
In this paper we improve the practical performance of flow-based methods for local clustering in two major ways. We first develop a generalized framework which takes better advantage of semi-supervised information about target clusters, and avoids the tendency of other methods to contract a large seed set into a small subcluster. Our approach allows users to place strict constraints and soft penalties on excluding specified seed nodes from the output set, depending on the user's level of confidence for whether or not each node should belong to the output set. 

Our second major contribution is a fast algorithm for minimizing our generalized objective function. We begin by showing that this objective can be minimized in strongly-local time using a meta-procedure that repeatedly solves localized minimum $s$-$t$ cuts, and doesn't require any explicit computation of maximum $s$-$t$ flows. This simultaneously generalizes and simplifies the meta-procedure we developed in previous work, which solves a more restricted objective function and requires the explicit computation of maximum flows as an intermediary step to obtaining minimum cuts~\cite{VeldtGleichMahoney2016}. We then implement our meta-procedure using a fast variant of the push-relabel algorithm~\cite{goldberg1988new}, which computes minimum cuts using preflows, rather than maximum flows. We make our algorithm extremely efficient using two key heuristics: a known global-relabeling scheme for the push-relabel algorithm~\cite{Cherkassky1997}, and a novel warm-start procedure which allows us to quickly solve consecutive minimum cut problems.

We validate our approach in several community detection experiments in real-world networks, and in several large-scale 3D image segmentation problems on graphs with hundreds of millions of edges. In practice our algorithm is faster that existing implementations of related flow-based methods, and allows us to more accurately detect ground truth clusters by better incorporating available knowledge of the target set.
\section{Background and Related Work}
We begin with an overview of notation and then provide a technical review of important concepts. Let $G = (V,E)$ represent an undirected and unweighted graph.\footnote{Following previous results, we will prove runtime and cut improvement guarantees for unweighted graphs, though in practice our implementations accommodate graphs with arbitrary floating point weights.}
For each node $v \in V$ let $d_v$ be its degree, i.e.\ the number of edges that have $v$ as an endpoint. For any set $S \subset V$ let $|E_S|$ be the number of interior edges in $S$, $\vol(S) = \sum_{v \in S} d_v$ be the \emph{volume} $S$, and $\cut(S) = \vol(S) - 2|E_S|$ denote the number of edges crossing from $S$ to $\bar{S} = V\backslash S$. Each set $S$ uniquely identifies a set of edges crossing from $S$ to $\bar{S}$, so we will frequently refer to a set of nodes $S$ as a \emph{cut} in a graph. One way to quantify the community structure of a set $S$ is by measuring its conductance:
\[ \phi(S) = \frac{\cut(S)} {\min\{ \vol(S), \vol(\bar{S})\}}. \]
A small value for $\phi(S)$ indicates that $S$ is well-connected internally but only loosely connected to the rest of the graph, and therefore represents a ``good'' cluster from a topological perspective.

\subsection{Local Variants of Conductance}
In local graph clustering we are given a seed (or \emph{reference}) set $R \subset V$ that is small with respect to the size of the graph. If we fix some value of a \emph{locality} parameter $\varepsilon \in \left[\frac{\vol(R)}{\vol(\bar{R})}, \infty \right)$, then the following objective function is a modification of the conductance score biased towards the set $R$, which we call the \emph{local conductance} measure:
\begin{equation}
\label{local-cond}
\phi_{R,\varepsilon}(S) = \frac{\cut(S)}{\vol(R\cap S) - \varepsilon \vol(\bar{R} \cap S)}. 
\end{equation}
One approach to localized community detection is to optimize the above objective over all sets $S$ such that the denominator is positive. This function was first introduced specifically for $\varepsilon = \frac{\vol(R)}{\vol(\bar{R})}$ by Andersen and Lang~\cite{AndersenLang2008}.
Orecchia and Zhou later considered larger values of $\varepsilon$, which effectively restricts the search space to sets $S$ that overlap significantly with $R$, leading to algorithms that minimize the objective in time independent of the size of the input graph $G$~\cite{OrecchiaZhu2014}. Both of these algorithms generalize the earlier Max-flow Quotient-cut Improvement (MQI) algorithm~\cite{LangRao2004}, which computes the minimum conductance subset of $R$, and fits the above paradigm if we allow $\varepsilon = \infty$.


\subsection{Minimizing Local Conductance}
\label{minlocalcond}
Although it is NP-hard to find the minimum conductance set of a graph $G$, one can minimize~\eqref{local-cond} efficiently by repeatedly solving a sequence of related minimum $s$-$t$ cut problems. First fix a parameter $\alpha \in (0,1)$ and construct a new graph $G_{st}$ using the following steps:
\begin{compactitem}
	\item Keep original nodes in $G$ and edges with weight 1
	\item Introduce source node $s$ and sink node $t$
	\item For every $r\in R$, connect $r$ to $s$ with weight $\alpha d_r$
	\item For every $j \in \bar{R}$, connect $j$ to $t$ with weight $\alpha \varepsilon d_j$.
\end{compactitem}
The minimum $s$-$t$ problem seeks the minimum weight set of edges in $G_{st}$ that, when removed, will separate the source node $s$ from the sink node $t$.
Every subset of nodes $S \subset V$ in $G$ induces an $s$-$t$ cut in $G_{st}$ where the two sides of the cut are $\{s\}\cup S$ and $\{t\} \cup \bar{S}$. The weight of this $s$-$t$ cut can be given entirely in terms of cuts and volumes of sets in $G$:
\begin{align*}
{\alg{STcut}}(S) &= \cut(S) + \alpha \varepsilon \vol(\bar{R} \cap S) + \alpha \vol(R \cap \bar{S})\\
&= \cut(S) + \alpha \varepsilon \vol(\bar{R} \cap S) - \alpha \vol(R \cap {S}) + \alpha \vol(R).
\end{align*}
If there exists some $S$ such that $\alg{STcut}(S) < \alpha \vol(R)$, one can show with a few steps of algebra that this implies $\phi_{R,\varepsilon}(S) < \alpha$. Therefore, $\phi_{R,\varepsilon}(S)$ can be minimized by finding the smallest $\alpha$ such that the minimum $s$-$t$ cut of $G_{st}$ is exactly $\alpha\vol(R)$. This can be accomplished by performing binary search over $\alpha$ or simply starting with $\alpha = \phi_{R,\varepsilon}(R)$ and iteratively finding minimum $s$-$t$ cuts in $G_{st}$ for increasingly smaller values of $\alpha$ until no more improvement is possible.
For sufficiently large $\varepsilon$, there is no need to explicitly construct $G_{st}$ in order to solve the min-cut objective. Instead, localized techniques can be used which repeatedly solve flow and cut problems on small subgraphs until a global solution is reached~\cite{OrecchiaZhu2014,VeldtGleichMahoney2016}. For more details on the flow-based framework presented here and its relationship to regularized optimization problems and random-walk based methods, we refer to related papers~\cite{Fountoulakis2017,pmlr-v32-gleich14}.

\subsection{Maximum $s$-$t$ Flows}
\label{maxflow}
Flow-based methods which operate by repeatedly solving minimum $s$-$t$ cut problems in the above manner include \alg{MQI}~\cite{LangRao2004}, \alg{FlowImprove}~\cite{AndersenLang2008}, \alg{LocalImprove}~\cite{OrecchiaZhu2014}, and \alg{SimpleLocal}~\cite{VeldtGleichMahoney2016}. These methods all obtain small $s$-$t$ cuts by solving the dual maximum $s$-$t$ flow problem. A major contribution in our work is to show that strongly-local algorithms for minimizing localized conductance measures can be obtained without any explicit computation of maximum $s$-$t$ flows. However, an efficient algorithm we develop later will apply explicit \emph{preflow} computations, so we briefly review key flow concepts here.

Let $G_A = (V\cup \{s,t\}, A)$ be a directed graph with a distinguished source and sink nodes $s$ and $t$ and capacities $c_{ij} > 0$ for each directed edge $(i,j)$ (called an \emph{arc}) in a set $A$. We can turn any undirected graph $G = (V\cup \{s,t\}, E)$ into a directed graph by replacing each edge $\{i,j\} \in E$ with two arcs $(i,j)$ and $(j,i)$. A valid $s$-$t$ flow on $G_A$ is a function $F : A \rightarrow \mathbb{R}_{\geq 0}$ which assigns flow values $f_{ij}$ satisfying 
\begin{align}
&  f_{ij} \leq c_{ij}  \text{ for $(i,j) \in A$} \\
& {\sum_{(j,i) \in A} f_{ji} = \sum_{(i,k) \in A} f_{ik} \text{ for $i \in V$} }
\end{align}
which are referred to as \emph{capacity} and \emph{flow} constraints respectively.
The flow $F$ is defined to be skew-symmetric, i.e. $f_{ij} = -f_{ji}$. The maximum $s$-$t$ flow problem seeks the flow $F$ which routes a maximum amount of flow from $s$ to $t$. Given a flow $F$ for a graph $G_A$, the residual graph $G_F = (V\cup\{s,t\}, A_F)$ is defined to be the directed network in which arc $(i,j) \in A_F$ has capacity $c_{ij}^F = c_{ij} - f_{ij}$. If an arc has nonzero residual capacity, it means that one can push more flow across it in search of new ways to route flow from $s$ to $t$. An arc $(i,j)$ is saturated if it has a residual capacity of zero. A flow $F$ is a maximum $s$-$t$ flow if and only if there exists no path of unsaturated arcs from $s$ to $t$. In this case, the set of nodes $S$ reachable from $s$ via a path of unsaturated arcs defines the minimum $s$-$t$ cut.

\subsection{Random Walks and Other Diffusion Based Clustering Algorithms}
Spectral methods are another widely popular approach to local graph clustering. Among these methods, the Andersen-Chung-Lang \alg{Push} procedure for computing an approximate personalized PageRank vector is well-known for its strongly-local runtime and good cut improvement guarantees~\cite{AndersenChungLang2006}. 
Random-walk based spectral methods typically find local cuts in a graph by running a localized diffusion from a small set of seed nodes. This diffusion produces an embedding with limited support over the nodes in the graph, which can then be rounded using some form of a sweet cut procedure to produce a cut. In contrast, flow-based methods solve biased minimum-cut computations and directly produce a cut rather than an embedding which must be rounded. 

Another key distinction between random-walk and flow-based approaches is the type of seed set these methods require. Random-walk diffusions are typically able to grow a single seed node or a small seed set into a larger cluster with good conductance. For example, Andersen et al.~\cite{AndersenChungLang2006} showed that if one starts from any one of a large number of individual seed nodes in a target cluster $T$, the \alg{Push} algorithm will return a localized cluster with conductance at most $O(\sqrt{\phi(T)})$. Flow-based methods are able to provide stronger cut improvement guarantees, but can only do so if they begin with a large seed set that has significant overlap with the target cluster $T$ (see e.g.\ the results in~\cite{AndersenLang2008,OrecchiaZhu2014,VeldtGleichMahoney2016}). In practice, flow-based methods may perform poorly if the seed set $R$ is too small. One approach for obtaining a large enough seed set is to first run a spectral diffusion from a small number of starting nodes and then refine the output using the flow-based method. Another approach is to take the starting seed nodes and grow them by a neighborhood with a small radius to produce a localized seed set that is sufficiently large for flow-based methods to output meaningful results.

\paragraph{Other Diffusion Based Methods}
In addition to random-walk based methods, there also exist other localized community detection algorithms that operate by computing a diffusion and then performing a sweep cut on the resulting embedding. Among others, Kloster and Gleich~\cite{Kloster-2014-hkrelax} developed a fast method for computing local communities based on the heat kernel diffusion. The runtime of their algorithm, \alg{hk-relax}, depends on the parameters of the diffusion but is independent of the size of the input graph; hence the method is strongly-local.
More recently, Wang et al.~\cite{pmlr-v70-wang17b} introduced the Capacity Releasing Diffusion (CRD), another strongly local algorithm, which spreads mass around nodes in a graph using a flow-like mechanism. Although CRD incorporates flow-based dynamics, we note that it does not compute biased minimum $s$-$t$ cuts on the input graph. For clarity, in this paper we reserve the term \emph{flow-based} to refer to methods that fit the paradigm outlined in Section~\ref{minlocalcond}.

\section{Generalized Local Clustering Objective}
In order to develop a flow-based method that places a higher emphasis on agreeing with the seed set, we begin by presenting a generalization of the local conductance objective~\eqref{local-cond}. After introducing the objective, we prove cut improvement guarantees that can be achieved if this objective is minimized in practice. In Section~\ref{sl-meta}, we prove that the objective can be solved in strongly-local time using a meta-procedure that repeatedly applies minimum $s$-$t$ cut solvers as subroutines with no explicit calculation of maximum $s$-$t$ flows. In Section~\ref{pr-implementation}, we provide details for how to implement the meta-procedure using a fast variant of push-relabel method with two key heuristics.

\subsection{The Seed-Penalized Conductance Score}
Let $G = (V,E)$ be an undirected and unweighted graph, and $R$ a small set of nodes that we wish to grow into a larger cluster that we will call $S$. Unlike other methods, we assume there exists a designated set of nodes $R_s \subseteq R$ which must be included in the output set, and a weight $p_i \geq 0$ for every other node $r_i \in R$ which indicates our level of confidence that $r_i$ should also be included in the output. We start by introducing the following new \emph{overlap score} between $R$ and $S$:
\begin{equation*}
\mathcal{O}_{R}(S) = \vol(R\cap S) - \varepsilon \vol(S\cap \barR) - \sum_{r \in R} p_r d_r \chi_{\barS}(r)
\end{equation*}
where $\vp = (p_i)$ is the vector of penalty weights for nodes in the seed set, $\chi_{\bar{S}}$ is the indicator function for nodes in $\bar{S}$, and $\varepsilon$ is a locality parameter controlling how much we allow the output set to include nodes outside $R$. The first term rewards a high intersection between $S$ and $R$, the second term penalizes the inclusion of nodes outside $R$, and the third term introduces a penalty for nodes in $R$ that are not in $S$. Given this definition of overlap, our goal is to minimize the following objective, which we refer to as seed-penalized conductance score:
\begin{equation}
\label{eq:sc}
\pi_R(S) = \begin{cases} \frac{\cut(S)}{\mathcal{O}_{R}(S)} & \mbox{ if } \mathcal{O}_{R}(S) > 0, R_s \subseteq S \\
\infty & \mbox{ otherwise}
\end{cases}
\end{equation}
To keep notation simple we only include the set $R$ in the subscript of $\mathcal{O}_R(S)$ and $\pi_R(S)$, though we note that these also depend on $R_s$, $\varepsilon$, and $\vp$, which are fixed parameters chosen at the outset of a problem.

\subsection{Cut Improvement Guarantee}
Despite the differences between~\eqref{eq:sc} and the standard conductance measure, we can prove that minimizing the former will give strong cut improvement guarantees in terms of the latter. This result is closely related to similar cut improvement guarantees for variants of the local conductance objective~\cite{AndersenLang2008,VeldtGleichMahoney2016}, but here we focus on the case where $R$ is completely contained in some ground truth target cluster $T$, since in our work we are especially concerned with approaches that grow a seed set into a larger cluster. Note that the following result is in fact independent of $R_s$ and $\vp$, so it holds regardless of how strict the seed node exclusion penalties are.
\begin{theorem}
	Let $G = (V,E)$ be connected and $R$ be a seed set. Let $T$ be any set of nodes containing $R$ with $\vol(T) \leq \vol(\bar{T})$, and assume that $\vol(R) = \gamma \vol(T)$ for some $\gamma \in (0,1)$. For any $\varepsilon \in \Big[ \frac{2 \vol(R)}{\vol(G) - 2\vol(R) }, \frac{\gamma}{1-\gamma}\Big)$, if $S^*$ is the set of nodes minimizing objective~\eqref{eq:sc}, then $\phi(S^*) \leq C \phi(T)$ where $C = \frac{1}{\gamma + \varepsilon \gamma - \varepsilon}$.
\end{theorem}
\begin{proof}
Note the following bound on the volume of $S^*$:
\begin{align*}
& 0 < \mathcal{O}_{R}(S^*) < \vol{(R \cap S^*)} - \varepsilon \vol{(\bar{R} \cap S^*)}\\
\implies & 0 < (1+\varepsilon) \vol{(R \cap S^*)} - \varepsilon \vol(S^*)\\
\implies & \vol(S^*) < \left( 1 + \frac{1}{\varepsilon}\right) \vol(R).
\end{align*}
Combining this with the lower bound on $\varepsilon$ given in the theorem statement, we see that the volume of $S^*$ is less than $\vol(G)/2$. Next observe that $\vol(S^*) > \mathcal{O}_{R}(S^*)$, so
\[\phi(S^*) = \cut(S^*)/\vol(S^*)<  \cut(S^*)/ \mathcal{O}_{R}(S^*) =\pi_R(S^*).\]
Because $R$ is contained in $T$, we have $\sum_{r \in R} p_r d_r \chi_{\bar{T}}(r) = 0$ and $\vol{(T \cap R)} = \vol(R)$. Therefore, 
\begin{align*}
\mathcal{O}_{R}(T) = \vol(R)-\varepsilon \vol{(T \cap \bar{R})} = (1+\varepsilon) \vol(R) - \varepsilon \vol(T) =  \vol(T)((1 + \varepsilon)\gamma -\varepsilon ).
\end{align*}
Since $S^*$ minimizes~\eqref{eq:sc},
\begin{align*}
\phi(S^*) &< \pi_R(S^*) \leq \pi_R(T) = \frac{\cut(T)}{\mathcal{O}_{R}(T)} =\frac{1}{(1+\varepsilon)\gamma - \varepsilon }  \frac{\cut(T)}{\vol(T)} = C \phi(T).
\end{align*} 
\end{proof} 
If we select $\varepsilon$ to be at its lower bound defined above, the approximation ratio will be $C = \frac{\vol(G) - 2\vol(R)}{\gamma \vol(G) + \gamma \vol(R) - 2\vol(R)}$. If $\vol(R)$ is very small compared to the overall size of the graph, then the approximation factor goes to $1/\gamma$ as the size of the graph increases for a fixed seed set.

\subsection{Minimizing Seed-Penalized Conductance}
\label{minspc}
As is the case for local conductance, objective~\eqref{eq:sc} can be minimized in polynomial time by solving a sequence of minimum $s$-$t$ cut problems. Fix $\alpha \in (0,1)$ and assume we wish to find whether there exists some $S$ such that $\pi_R(S) < \alpha$. We construct a new version of the \emph{cut graph} $G_{st}$ which includes all nodes in $G$ and an additional source $s$ and sink $t$. For any node $r \in R_s$, we will assign a penalty variable $p_r = \vol(G)/\alpha$. We add an edge from $s$ to each $r \in R$ with weight $\alpha(1+p_r)d_r$. The chosen weight for nodes in $R_s$ is large enough to guarantee that a minimum cut will never separate any of these nodes from $s$. Then, for each node $w \in \bar{R}$, we add an edge from $w$ to $t$ with weight $\alpha \varepsilon d_w$. 

For any set of non-terminal nodes $S \subseteq V$, the $s$-$t$ cut associated with that set can be expressed in terms of cuts and volumes in the original graph $G$:
\begin{equation*}
{\textstyle \cut(S) + \alpha\varepsilon \vol(\bar{R} \cap {S}) + \alpha \sum_{r\in R} d_r (1+p_r) \chi_{\barS}(r)}. \notag 
\end{equation*}
Using the observation that $\alpha\vol(R \cap \bar{S}) = \alpha \vol(R) - \alpha \vol(R\cap S)$, we can rearrange this into the following objective function:
\begin{equation}
\label{mincut}
f_{R,\varepsilon}^\alpha(S) = \cut(S) - \alpha \mathcal{O}_{R}(S) + \alpha \vol(R),
\end{equation}
so $f_{R,\varepsilon}(S) < \alpha \vol(R)$ if and only if $\cut(S)/\mathcal{O}_{R}(S) < \alpha$. Thus, solving the minimum $s$-$t$ cut objective on $G_{st}$ will tell us whether there exists some $S$ with seed-penalized conductance less than $\alpha$.


Given any procedure for minimizing objective~\eqref{mincut} (e.g.\ a generic minimum $s$-$t$ cut solver), we can minimize seed-penalized local conductance using Algorithm~\ref{alg:wrapper}.
\begin{algorithm}[t]
	\caption{Minimizing seed-penalized conductance }
	\label{alg:wrapper}
	\begin{algorithmic}
		\STATE \textbf{Input:} $G$, $R$, $\varepsilon$, $\vp$
		\STATE $\alpha := 2$
		\STATE $\alpha_{new} = \pi_R(R) = \phi(R)$
		\STATE $S = R$
		\WHILE{$\alpha_{new} < \alpha$}
		\STATE $S_{best} \leftarrow S$
		\STATE $\alpha \leftarrow \alpha_{new}$
		\STATE $S \leftarrow \arg \min f_{\varepsilon,R}^\alpha(S)$
		\STATE $\alpha_{new} \leftarrow \pi_R(S)$
		\ENDWHILE
		\STATE \textbf{Return:} $S_{best}$
	\end{algorithmic}
\end{algorithm}
%
We end this section by showing a bound on the number of iterations for Algorithm~\ref{alg:wrapper}, by slightly adapting the techniques Andersen and Lang~\cite{AndersenLang2008} used to prove a similar bound for a more restrictive objective function.
\begin{theorem}
	\label{thm:cutR}
	Algorithm~\ref{alg:wrapper} will need to solve min-cut objective~\eqref{mincut} at most $\cut(R)$ times.
\end{theorem}
\begin{proof}
Since $R$ and $\varepsilon$ and $\vp$ are fixed at the outset of the algorithm we will write $f^\alpha$ instead of $f_{R,\varepsilon}^\alpha$ and $\mathcal{O}$ rather than $\mathcal{O}_{R}$. Consider two consecutive iterations in which Algorithm~\ref{alg:wrapper} successfully finds sets with improved seed-penalized conductance and therefore doesn't terminate. Let $S_i$ be the set returned after the $(i-1)$st iteration, so $S_{i} = \arg\min f^{\alpha_{i-1}}(S)$ for some $\alpha_{i-1}$, and set $\alpha_i = \pi_R(S_i) = \cut(S_i)/ \mathcal{O}(S_i) < \alpha_{i-1}$. Similarly, $S_{i+1} = \arg\min f^{\alpha_{i}}(S)$ and $\alpha_{i+1} = \pi_R(S_{i+1}) < \alpha_i$.
Note that
\begin{align*}
f^{\alpha_{i-1}}(S_i) &= \alpha_{i-1}\vol(R) + \cut(S_i) - \alpha_{i-1} \mathcal{O}(S_i)\\
							&= \alpha_{i-1}\vol(R) + \mathcal{O}(S_i) ( \pi_R(S_i) - \alpha_{i-1}) \\
							&= \alpha_{i-1}\vol(R) + \mathcal{O}(S_i) (\alpha_i - \alpha_{i-1})
\end{align*}
and similarly
\[ f^{\alpha_{i-1}}(S_{i+1}) = \alpha_{i-1}\vol(R) + \mathcal{O}(S_{i+1}) (\alpha_{i+1} - \alpha_{i-1}).  \]
Because $S_{i}$ minimizes $f^{\alpha_{i-1}}$ we know that $f^{\alpha_{i-1}} (S_i) \leq f^{\alpha_{i-1}} (S_{i+1})$, which implies that
\[ \mathcal{O}(S_{i})(\alpha_{i} - \alpha_{i-1}) \leq \mathcal{O}(S_{i+1})(\alpha_{i+1} - \alpha_{i-1})\]
and since $(\alpha_{i+1} - \alpha_{i-1}) < (\alpha_{i} - \alpha_{i-1}) < 0$ we see that $\mathcal{O}(S_{i+1}) < \mathcal{O}(S_i)$. Thus both $\pi_R(S)$ and its denominator are strictly decreasing during the course of the algorithm, so $\cut(R)$ must strictly decrease at each step as well. Since we assume the graph is unweighted, there are at most $\cut(R)$ iterations in total.
\end{proof} 

\section{The Strongly-Local Meta-Procedure}
\label{sl-meta}
The results of the previous section imply that Algorithm~\ref{alg:wrapper} can be run in polynomial time using any black-box min $s$-$t$ cut solver. In this section we will prove a much stronger result by showing that objective~\eqref{mincut} can be minimized in strongly-local time using a very simple two-step meta-procedure. A significant feature of this meta-procedure is that the algorithm does not require any explicit computation of maximum flows, but relies on a very simple repeated two-step procedure for localized minimum $s$-$t$ cuts. 

\paragraph{Local Graph Operations.} In order minimize~\eqref{mincut} without touching all of $G = (V,E)$, we will repeatedly solve a variant of objective~\eqref{mincut} on a growing subgraph $L = (V,E_L)$ called the \emph{local graph}, which contains a restricted edge set $E_L \subset E$. In theory $L$ is assumed to have the same node set $V$ but many of these nodes will have degree zero in $L$,  so we will not need to explicitly perform computations with all nodes in practice.

We consider the following localized variant of~\eqref{mincut}, which corresponds to a minimum $s$-$t$ cut problem on a subgraph $L_{st}$ of the cut graph $G_{st}$:
\begin{equation}
\label{localmincut}
f_{L}^\alpha(S) = \cut_L(S) - \alpha \mathcal{O}_{R}(S) + \alpha \vol(R),
\end{equation}
where the only difference from~\eqref{mincut} is that $\cut_L(S)$ is defined to the be number of edges in $E_L$ between $S$ and $\bar{S}$, rather than the number of edges in $E$, thus $f_L^\alpha(S) \leq f_G(S) = f_{R,\varepsilon}^\alpha(S)$ for all $S \subset V$ and for any such subgraph $L$ of $G$. We will use the notation $d_i^L$ to denote the degree of node $i$ in $L$, which is always less than or equal to $d_i$. We then distinguish between an \emph{edge-complete} set of nodes $L_C= \{i \in V : d_i = d_i^L \} $ and an \emph{edge-incomplete} nodes $L_I = \{i \in V : d_i > d_i^L \}$ in $L$.

Let $S_L$ be the minimizer of~\eqref{localmincut} for a fixed subgraph $L$. The following lemma shows that if $S_L$ is made up entirely of edge-complete nodes, then this set also minimizes the global objective function $f_G = f_{R,\varepsilon}^\alpha$~\eqref{mincut}.
\begin{lemma}
	\label{lem:local}
	Let $S_L = \arg \min f_L(S)$. If $S_L \subseteq L_C$ then $S_L = \arg \min f_G(S)$.
\end{lemma}
\begin{proof}
	Because $E_L \subseteq E$, $\cut_L(S) \leq \cut(S)$ for all $S\subset V$, and therefore $f_L(S) \leq f_G(S)$ for all $S\subset V$, which implies that $\min_S \, f_L(S) \leq \min_S \, f_G(S)$. For the specified set $S_L$, since $S_L \subseteq L_C$, all nodes in $S_L$ have the same degree in $L$ as well as $G$, implying that $\cut_L(S_L) = \cut(S_L)$. Therefore:
	\[ f_L(S_L) = f_G(S_L) \geq \min_S f_G(S) \geq \min_S f_L(S) = f_L(S_L) \]
	so equality holds throughout and $S_L$ is optimal for both $f_L$ and $f_G$.
\end{proof}
Our meta-procedure for minimizing~\eqref{mincut} in strongly local time operates by repeatedly solving objective~\eqref{localmincut} over a sequence of growing local subgraphs. 
This proceeds until an iteration in which the current subgraph $L$ is large enough so that the set minimizing~\eqref{localmincut} is made up of edge-complete nodes, at which point we know by Lemma~\ref{lem:local} that we have globally solved objective~\eqref{mincut}. The full procedure is given in Algorithm~\ref{alg:meta}.
\begin{algorithm}
	\caption{\alg{Local Min-Cut Meta-Procedure}}
	\label{alg:meta}
	\begin{algorithmic}
		\STATE {\bfseries Input:} graph $G$, seed set $R$, parameters $\alpha, \varepsilon$, $\vp$
		\STATE Initialize $L$: $L_C = R$, $L_I = \bar{R}$
		\STATE $E_L$: all edges in $E$ with at least 1 endpoint in $R$.
		\REPEAT
		\STATE \textbf{1. Solve Local Objective on $L$} 
		\STATE $S_L = \arg\min_S \, f_L^\alpha(S)$
		\STATE $N = S_L\cap L_I$ (new nodes to explore around)
		\STATE \textbf{2. Expand $L$ around $N$}
		\FORALL{$v \in N$}
		\STATE $E_v = \text{edges incident to node $v$ in $G$}$
		\STATE $E_L \leftarrow E_L \cup E_v$
		\STATE $L_C \leftarrow L_C \cup\{v\}$
		\STATE $L_I \leftarrow L_I - v$.
		\ENDFOR
		\STATE $L \leftarrow (V,E_L)$
		\UNTIL{$N = \emptyset$}
	\end{algorithmic}
\end{algorithm}

The following result proves that the size of the largest subgraph formed by Algorithm~\ref{alg:meta} will be bounded in terms of $\vol(R)$, and $\varepsilon$. This mirrors a result for our previously developed meta-procedure~\cite{VeldtGleichMahoney2016}, which required explicit computation of flows in order to solve an objective related to~\eqref{mincut}. The proof technique is very similar, though in order to avoid explicit computation of flows, more analysis is needed to prove the theoretical bound.
\begin{theorem}
	\label{thm:volbound}
	Let $\alpha$ be chosen so that $\pi_R(S_0) = \alpha$ for some $S_0 \subset V$. The largest subgraph $L$ formed by Algorithm~\ref{alg:meta} satisfies the following volume bound:
	\[ \vol(L) = 2|E_L| \leq \vol(R) \left(1 + 1/\varepsilon\right) + \cut(R). \]
\end{theorem}
\begin{proof}

	Recall that the global objective~\eqref{mincut} we wish to solve is exactly the minimum $s$-$t$ objective on an auxiliary graph $G_{st}$, whose construction we outlined in Section~\ref{minspc}. The localized objective~\eqref{localmincut} is the minimum $s$-$t$ cut objective on a subgraph $L_{st}$ of $G_{st}$ that contains the same set of terminal edges but only a subset of edges between non-terminal nodes. Although Algorithm~\ref{alg:meta} does not require an explicit computation of a maximum $s$-$t$ flow, we will show a bound on $\vol(L)$ by considering an implicit maximum $s$-$t$ flow with special properties that exists by the max-flow/min-cut duality theorem. 
	
	\paragraph{Notation for Proof.} Let $L^{(i)}$ denote the local graph at the $i$th iteration of Algorithm~\ref{alg:meta}, $f_i = f_{L_i}^\alpha$ be shorthand for objective function~\eqref{localmincut}, and $S_i = \arg\min f_i(S)$. Use $N_i$ to denote the set of nodes which become edge-complete in the $i$th iteration, which by design are all in $\bar{R}$. $L^{(i)}$ itself is a subgraph of $G$, and we use $L_{st}^{(i)}$ to denote the subgraph of $G_{st}$ whose minimum $s$-$t$ cut we compute at iteration $i$. Let $C^{(i)}$ denote the set of edges in $L_{st}^{(i)}$ that are cut at iteration $i$. 
	
	
	\paragraph{Constructing Implicit Flows.} By the min-cut/max-flow theorem, the value of the maximum flow on the graph $L_{st}^{(i)}$ equals the weight of the cut $C^{(i)}$ which we compute in practice. Furthermore, any maximum flow which we could compute will saturate all of the edges in $C^{(i)}$. Using this observation we will show by construction that when Algorithm~\ref{alg:meta} terminates after some iteration $k$, there exists a maximum flow $F$ on $L_{st}^{(k)}$ which saturates all edges between edge-complete nodes and the sink. 
	
	In the first iteration, $N_1$ is exactly the set of nodes whose edge to the sink is cut. Let $F_1 = (f_{ij})$ be some maximum $s$-$t$ flow on $L_{st}^{(1)}$, and note that it will saturate all edges from $N_1$ to $t$. In the next iteration we compute a new minimum cut $C^{(2)}$. The set $N_2$ represents all nodes whose edge to $t$ was included in $C^{(2)}$ but not $C^{(1)}$. Since $L_{st}^{(1)}$ is a strict subgraph of $L_{st}^{(2)}$, we could in theory find a maximum $s$-$t$ flow $F_2$ on $L_{st}^{(2)}$ by starting with the previous flow $F_1$ and continually finding new augmenting flow paths until no more flow can be routed from $s$ to $t$. Since $F_2$ is a maximum flow, it must saturate all edges from $N_2$ to $t$, since these were cut by $C^{(2)}$. Furthermore, we can assume that in the construction of $F_2$, all the edges from $N_1$ to $t$, which were saturated by $F_1$, will remain saturated, since no improvement can be gained by rerouting flow from the sink back to $N_1$. Proceeding by induction, we see that at iteration $i$ we can find some maximum $s$-$t$ flow which saturates all edges from $N_i$ to the sink (since these were cut by $C^{(i)}$), and also saturates all terminal edges of previous edge-complete nodes in $\bar{R}$. We conclude that when Algorithm~\ref{alg:meta} terminates, the maximum flow, and hence the minimum cut, will be bounded below by the weight of edges from edge-complete nodes to $t$. Each edge is of weight $\alpha \varepsilon d_v$ for the edge-complete node $v \in \bar{R}$. Thus $ \alpha \varepsilon \vol(\mathcal{N}) < M$ where $M$ is the minimum cut value, and $\mathcal{N}$ is the set of edge-complete nodes.
	
	\paragraph{Bounding Min-Cut Above.} Next we bound $M$ from above. In the statement of the theorem we assumed that that $\pi_R(S_0) = \alpha$ for some $S_0 \subset V$ (which will always be true if we use Algorithm~\ref{alg:meta} as a subroutine for Algorithm~\ref{alg:wrapper}). This is equivalent to the statement that
	\[ f_{R,\varepsilon}^\alpha(S_0) = \cut(S_0) - \alpha \mathcal{O}_{R}(S_0) + \alpha \vol(R) = \alpha \vol(R), \]
	so we have an upper bound of $\alpha\vol(R)$ on $\min_S f_{R,\varepsilon}^\alpha(S)$. Combining upper and lower bounds, we see that
	\[\alpha \varepsilon \vol(\mathcal{N}) \leq \alpha \vol(R) \implies \vol(\mathcal{N}) \leq \vol(R)/\varepsilon. \]

	\paragraph{Bounding Volume of $L$.} The largest local graph $L$ that we form is made up of $\caln$, all of $R$, and a few additional nodes in $\bar{R}$ that have non-zero degree in $L$ but remained edge-incomplete during the entire course of the algorithm. Use $P$ to denote these edge-incomplete nodes, and note that they share no edges with each other, but only share edges with $R$ and $\caln$. Thus, 
	\begin{align*}
	 \vol(P) &\leq (\text{number of edges from $R$ to $P$})  + (\text{number of edges from $\caln$ to $P$}) \\
				 & \leq \cut(R) + \vol(\caln).
	\end{align*}
	Thus the full volume bound follows:
	\begin{align*}
	\vol(L) &\leq \vol(R) + \vol(\caln) + \vol(P) \\
		 &\leq \vol(R) + 2 \vol(\caln) + \cut(R) \\
		&  \leq \vol(R) ( 1 + 2/\varepsilon) + \cut(R). 
	\end{align*}
\end{proof}

The volume bound given here is the same as the bound shown for our previous method \alg{SimpleLocal}. Thus, using Algorithm~\ref{alg:meta} as a subroutine in Algorithm~\ref{alg:wrapper} produces an algorithm with the same theoretical runtime as \alg{SimpleLocal}, despite solving a much more general objective function and completely avoiding any explicit maximum flow computations. Each flow problem takes at most $O(\vol(R)^3/\varepsilon)$ operations if fast flow subroutines are used~\cite{Orlin-2013-max-flow}, and Theorem~\ref{thm:cutR} guarantees it will be run at most $\cut(R)$ times. This runtime bound is very conservative and the empirical performance will typically be significantly better.

\section{The Push-Relabel Implementation}
\label{pr-implementation}
We implement Algorithm~\ref{alg:meta} using a new method for computing minimum $s$-$t$ cuts based on a variant of the push-relabel algorithm of Goldberg and Tarjan~\cite{goldberg1988new}. The full push-relabel algorithm can be separated into two phases: the first phase computes a maximum preflow which can be used to solve the minimum $s$-$t$ cut problem, and the second phase performs additional computations to turn the preflow into a maximum $s$-$t$ flow. Because we only require minimum $s$-$t$ cuts, our method simply applies Phase 1. 

\subsection{Push-Relabel Overview}
Section~\ref{maxflow} provides a basic overview of flow computations. The push-relabel algorithm is specifically a \emph{preflow} algorithm for maximum flows, meaning that during the course of the algorithm, all arcs satisfy capacity constraints, but each node $i$ is allowed to have more incoming flow than outgoing flow, i.e.\ a preflow satisfies a relaxation of the the flow constraints:
\begin{align*}
{ \sum_{(j,i) \in A} f_{ji} \leq \sum_{(i,k) \in A} f_{ik} \text{ for $i \in V$}}
\end{align*}
where $F = (f_{ij})$ is a flow assignment for a directed graph $G$ with node set $V$ and arc set $A$. Push-relabel maintains a labeling function $\ell : V \rightarrow \{0,1,2, \hdots , n \}$ where $n = |V|$ is the number of nodes in a graph $G_{st}$ with distinguished source and sink. 
%
The algorithm can be initialized using any preflow and a labeling that gives a lower bound on the distance from each node to the sink in the residual graph. The standard initialization is to set $\ell(s) = n$ and the label of all other nodes to zero. The preflow is initialized to be zero on all edges, and afterwards all edges from $s$ to its neighbors are saturated. This creates a positive \emph{excess} at these neighbors, i.e.\ more flow goes into the nodes than out. After initialization, the algorithm repeatedly visits \emph{active} nodes, which are nodes that have a label less than $n$ and a positive excess. For a selected active node $u$, the algorithm locally pushes flow across admissible edges, which are defined to be edges $(u,v)$ for which $\ell(u) = \ell(v) + 1.$ If no admissible edges exist, the label of the node is increased to be the minimum label such that an admissible arc is created. During the course of the algorithm, it can be shown that $\ell(u) < \ell(v)$ for any arc $(u,v)$ with nonzero residual capacity, and furthermore $\ell(v)$ is a lower bound on the distance from node $v$ to the sink $t$, if there still exists a path of unsaturated edges from $v$ to $t$. Phase 1 of the algorithm is complete when there are no more active nodes to process. At this point the preflow is at a maximum, and the set of nodes with label $n$ forms the minimum cut set.

\subsection{Label Selection Variants and Relabeling Heuristics}
The generic push-relabel algorithm simply requires one to push flow across admissible edges whenever there still exist active nodes. This procedure is guaranteed to converge to the solution to the minimum cut problem, but better runtimes can be obtained by more carefully selecting the order in which to process active nodes. One approach is the first-in-first-out (FIFO) method, which begins by pushing all initial active nodes into a queue, and adding new nodes to the queue as they become active. Another approach is to continually select the highest-labeled node at each step. 

The push-relabel method can be made very fast in practice using efficient relabeling heuristics~\cite{Cherkassky1997}. One simple but very effective heuristic is to periodically run a breadth first search from the sink node $t$ and update the labels of each node to equal the distance from that node to $t$. Another heuristic is the gap relabeling heuristic, which checks whether there exist certain types of gaps in the labels that can be used to prove when certain nodes are no longer connected to the sink node $t$.

\subsection{Implementation Details and Warm-Start Heuristic}
In practice we implement the FIFO push-relabel algorithm in the Julia programming language and make use of the global relabeling heuristic. Although implementations of push-relabel in other languages have made efficient use of the highest-label variant and the gap relabeling heuristic~\cite{Cherkassky1997}, these require slightly more sophisticated data structures that are more challenging to maintain in Julia. Our implementation choices make it possible to maintain a very simple but efficient method to implement Algorithm~\ref{alg:meta}. Running this procedure for various $\alpha$ using Algorithm~\ref{alg:wrapper} provides a fast local graph clustering algorithm. Because our method is flow-based and puts a higher emphasis on including seed nodes, we refer to it as \alg{FlowSeed}.

An important part of our implementation of Algorithm~\ref{alg:meta} is a warm-start heuristic for computing consecutive minimum $s$-$t$ cuts on the growing subgraph. Each local subgraph $L$ corresponds to a local cut graph $L_{st}$ with added source and sink nodes. For the first local cut graph, we use the standard initialization for push-relabel, i.e. start with a preflow of zero and saturate all edges from $s$ to its neighbors. Applying push-relabel will return a maximum preflow $F$ on $L_{st}$, and thus a minimum $s$-$t$ cut which we use to update $L$ as outlined in Section~\ref{sl-meta}. After $L$ and $L_{st}$ are updated, the goal is to find an updated minimum cut, which can be accomplished up finding an updated maximum preflow. Note that $F$ is no longer a maximum preflow on the updated $L_{st}$, since we have added new nodes and edges to $L_{st}$ and hence there are new ways to route flow from $s$ to $t$. However, $F$ will still be a valid preflow. Our warm-start procedure therefore initializes the next run of the push-relabel method with the preflow $F$, and sets the label of each node to be its distance to the sink in the corresponding residual graph. Initializing each consecutive maximum preflow computation in this way will be much more efficient than re-constructing $L_{st}$ from $L$ at each step and starting with a preflow of zero.
\section{Experiments}
We demonstrate the performance of \alg{FlowSeed} in several community detection experiments and large scale 3D image segmentation problems. Code for our method and experiments is available online at~\url{https://github.com/nveldt/FlowSeed}.

%
%
%
%

\subsection{Local Community Detection}
\label{cd}
Our first experiment demonstrates the robustness of \alg{FlowSeed} in local community detection, thanks to its ability to penalize the exclusion of certain seed nodes from the output. 

\paragraph{Datasets}
We consider four graphs from the SNAP repository~\cite{snapnets}: DBLP, Amazon, LiveJournal, and Orkut. Each network come with sets of nodes representing so-called ``functional communities''~\cite{Yang2015}. Communities in these networks specifically represent user groups in a social network (LiveJournal and Orkut), product categories (Amazon), or academic publication venues (DBLP). For each graph we select the ten largest communities out of the top 5000 communities identified by Yang and Leskovec~\cite{yangComm}, which come with the data on the SNAP website. These communities range in size from a few hundred to a few thousand nodes. The size of each network in terms of nodes and edges is given in Table~\ref{snap-stats}, along with the average community size and conductance among the ten largest communities in each network. 
\begin{table}[h]
	\caption{Number of nodes and edges for SNAP networks, along with target community size $|T|$ and target community conductance $\phi(T)$, averaged over the largest 10 communities.}
	\label{snap-stats}
	\centering
	\begin{tabular}{lllll}
		\toprule
		Graph & $|V|$  & $|E|$ & $|T|$ & $\phi(T)$  \\
		\midrule
 DBLP & 317,080 & 1,049,866 & 3902 & 0.4948 \\
 Amazon & 334,863 & 925,872 & 190 & 0.0289 \\
 LiveJournal & 3,997,962 & 34,681,189 & 988 & 0.4469 \\
 Orkut & 3,072,441 & 117,185,083 & 3877 & 0.6512 \\
		\bottomrule
	\end{tabular}
\end{table} 

\paragraph{Strongly-Local Algorithms}
We compare our Julia implementation of \alg{FlowSeed} against several other standard local graph clustering algorithms that come with strong locality guarantees:\newline

\noindent
\alg{Push}: The random-walk diffusion method of Andersen et al.~\cite{AndersenChungLang2006}. We use a highly optimized C++ implementation of the algorithm with a MATLAB interface. This method relies on a PageRank teleportation parameter $\alpha_{pr}$ and a tolerance parameter $\varepsilon_{pr}$. The latter controls the accuracy to which the underlying PageRank problem has been solved, and implicitly controls how wide of a region is explored in the graph by the method.\newline 

\noindent
\alg{HK-relax}: The heat kernel diffusion method of Kloster and Gleich~\cite{Kloster-2014-hkrelax}. This comes with diffusion parameter $t$ and a tolerance parameter $\varepsilon_{hk}$. We use the C++ implementation (with MATLAB interface) provided by the original authors, available online at~\url{https://github.com/kkloste/hkgrow}. \newline

\noindent 
\alg{CRD}: The capacity releasing diffusion of Wang et al.~\cite{pmlr-v70-wang17b}, implemented as a part of the LocalGraphClustering package~\url{https://github.com/kfoynt/LocalGraphClustering}. For this method we must set parameters $U$, $h$, and $w$ (see the original work for details). \newline

\noindent
\alg{SimpleLocal}: Our previous strongly-local flow-based method which optimizes the localized conductance objective~\eqref{local-cond} for a seed set $R$ and locality parameter $\varepsilon$. We use the fast C++ implementation available from the LocalGraphClustering package. \newline

\paragraph{Seed Set and Algorithm Parameters}
For each target community in each network, we randomly select 5\% of the target nodes, which we refer to as the \emph{starter} nodes. We grow the starter nodes by their neighborhood to produce a seed set $R$ that we use as input for each algorithm. For \alg{HK-relax}, \alg{Push}, and \alg{CRD}, we also tried using the starter set as the seed set, but this was not as effective in practice. Similarly, for these three methods we tried using each one of the starter nodes one at a time as an individual seed node, taking the best conductance output as the result, but this also was ineffective. Therefore, we only report results for each method using the full seed set $R$.

For both \alg{SimpleLocal} and \alg{FlowSeed} we use a locality parameter of $\varepsilon = 0.1$. We require \alg{FlowSeed} to strictly include the known 5\% of the target community, but do not add any soft penalties on excluding other seed nodes. For \alg{Push}, we set $\alpha_{pr} = 0.99$, and test a range of tolerance parameters: $\varepsilon_{pr} \in \{10^{-2}, 10^{-3}, 10^{-4}, 10^{-5}, 10^{-6}, 10^{-7} \}$, returning the output with the best conductance. For \alg{HK-relax}, following the experiments of Kloster and Gleich~\cite{Kloster-2014-hkrelax}, we test several pairs of parameter settings: $(t,\varepsilon_{hk})$ = $(5, 10^{-4})$; $(10, 10^{-4} )$; $(20, 10^{-3} )$; $(40, 10^{-3} )$; $(80, 10^{-2} )$, again returning the lowest conductance output. We also experimented with smaller values for locality and tolerance parameters~$\varepsilon$, $\varepsilon_{pr}$, and $\varepsilon_{hk}$, all of which control how much of the graph is explored by their respective method. However, this only increased the runtime of these methods without yielding improved results. This is consistent with previous research that has shown that real-world networks often exhibit very good community structure at small scales, but do not tend to possess large sets with good topological community structure~\cite{leskovec2008statistical}. Thus exploring the graph using \alg{SimpleLocal}, \alg{FlowSeed}, \alg{HK-relax}, or \alg{Push} with a smaller tolerance or locality parameter is unlikely to return better results. Finally, for \alg{CRD} we use the default parameters $U = 3$, $h = 10$, and $w = 2$; we were unable to determine other parameter settings that led to consistently improved output in practice. 

\paragraph{Results}
In Table~\ref{tab:snap}, we report conductance, cluster size, precision, recall, F1 scores, and runtimes for each method, averaged over the 10 communities for each network. \alg{FlowSeed} returns the best result among all methods on three of four datasets. Furthermore, it always outperforms \alg{SimpleLocal}, which solves a very similar objective but does not penalize the exclusion of seed nodes. The relative performance of all methods other than \alg{FlowSeed} varies significantly depending on the dataset. As expected, in many cases \alg{SimpleLocal} discards too many seed nodes, returning sets with very good conductance and precision, at the expense of poor recall. This is exhibited most clearly on the DBLP dataset, and to a lesser extent on Amazon. On DBLP, \alg{HK-relax} and \alg{Push} also exhibit a tendency to overemphasize conductance and output tiny sets with low recall. On Orkut, \alg{Push} grows sets that are too large, which is another tendency of the method that has been documented in other work as well~\cite{Kloster-2014-hkrelax,VeldtGleichMahoney2016}. \alg{CRD} outperforms all other methods on LiveJournal and does reasonably well (relative to the other methods) on both Orkut and DBLP. However, it returns results that are significantly worse than all other algorithms on Amazon.
\begin{table}[h]
	\caption{Average set size, conductance $\phi$, runtime (in seconds), precision, recall, and F1 score for five methods on four networks. Best F1 scores are shown in bold.}
	\label{tab:snap}
	\vspace{-.5\baselineskip}
	\centering
	\begin{tabular}{llllllll}
		\toprule
		Graph &  method & size & $\phi$ & runtime & prec. & recall & F1 \\
 \midrule 
 DBLP & \alg{HK-relax} & 280 & 0.100 & 0.110 & 0.609 & 0.036 & 0.044 \\
 & \alg{Push} & 80 & 0.130 & 0.168 & 0.607& 0.011 & 0.022 \\
 & \alg{CRD}  & 1460 & 0.255 & 3.569 & 0.468& 0.190 & 0.263 \\
 & \alg{SimpleLocal} & 31 & 0.046 & 24.540 & 0.632& 0.006 & 0.011 \\
 & \alg{FlowSeed} & 2789 & 0.254 & 9.491 & 0.414& 0.302 &  \textbf{0.345} \\
 \midrule 
 Amazon & \alg{HK-relax} & 156 & 0.007 & 0.020 & 0.952 & 0.804 & 0.843 \\
 & \alg{Push} & 180 & 0.007 & 0.225 & 0.953& 0.889 & 0.904 \\
 & \alg{CRD}  & 70 & 0.208 & 1.629 & 0.958& 0.374 & 0.521 \\
 & \alg{SimpleLocal} & 154 & 0.007 & 0.096 & 0.906& 0.772 & 0.814 \\
 & \alg{FlowSeed} & 214 & 0.018 & 0.332 & 0.892& 0.970 &  \textbf{0.924} \\
 \midrule 
 LiveJournal & \alg{HK-relax} & 1373 & 0.144 & 0.206 & 0.432 & 0.593 & 0.406 \\
 & \alg{Push} & 1867 & 0.363 & 0.172 & 0.444& 0.650 & 0.489 \\
 & \alg{CRD}  & 3230 & 0.098 & 69.584 & 0.464& 0.782 & \textbf{0.520} \\
 & \alg{SimpleLocal} & 4485 & 0.035 & 17.932 & 0.371& 0.813 & 0.440 \\
 & \alg{FlowSeed} & 4931 & 0.070 & 22.780 & 0.395& 0.896 & 0.484 \\
 \midrule 
 Orkut & \alg{HK-relax} & 3540 & 0.648 & 3.748 & 0.103 & 0.198 & 0.084 \\
 & \alg{Push} & 16790 & 0.749 & 0.767 & 0.165& 0.706 & 0.267 \\
 & \alg{CRD}  & 4006 & 0.355 & 451.092 & 0.442& 0.457 & 0.428 \\
 & \alg{SimpleLocal} & 3726 & 0.339 & 327.118 & 0.468& 0.451 & 0.439 \\
 & \alg{FlowSeed} & 4049 & 0.379 & 439.327 & 0.505& 0.507 & \textbf{0.494} \\
		\bottomrule
	\end{tabular}
\end{table}

\paragraph{Runtime Comparison of Flow-Based Methods}
In Table~\ref{tab:snap} we see that \alg{Push} and \alg{HK-relax} are by far the fastest local clustering algorithms, taking only a fraction of a second in almost all cases. However, although these methods sometimes return good outputs, they do not consistently perform well across all datasets. Focusing next on the two flow-based methods, we see that \alg{SimpleLocal} and \alg{FlowSeed} trade off in runtime for the experiments summarized in Table~\ref{tab:snap}. However, the difference in runtime is greatly influenced by the fact that \alg{FlowSeed} solves a slightly different objective in order to ensure certain seed nodes are included in the output. In order to provide a clearer runtime comparison between these two related methods, we re-run both algorithms again but do not include any seed exclusion penalties when running \alg{FlowSeed}. In this case the algorithms will solve the same exact objective and return the same output. This time we use a locality parameter that depends on the size of the seed set relative to the graph: $\varepsilon =5\vol(R)/\vol(\bar{R})$. This means that for the larger datasets, computations will not be as local. Therefore, the bottleneck for both of the methods will be their underlying flow subroutines, which are what we are most interested in comparing. The average runtimes for the two algorithms are given in Table~\ref{tab:runtimes}.
%
\begin{table}[t]
	\caption{Average runtimes (in seconds) for \alg{SimpleLocal} and \alg{FlowSeed} with no seed exclusion penalties. In this case the problems optimize the same objective, but \alg{FlowSeed} is faster thanks to our warm-start heuristic and push-relabel flow subroutine. }
	\label{tab:runtimes}
	\centering
\begin{tabular}{l c c c c}
& DBLP & Amazon & LiveJoural & Orkut \\
\toprule
\alg{FlowSeed}&5.4 & 1.3 & 107.8 & 229.3 \\
\alg{SimpleLocal} & 17.6& 3.5 & 134.9  & 632.2 \\
\bottomrule
\end{tabular}
\end{table}
From these results we see that, thanks to our push-relabel implementation and warm start procedure, our Julia implementation outperforms the optimized C++ code for \alg{SimpleLocal}, which relies on Dinic’s maximum flow algorithm as a subroutine and makes no use of warm starts. Thus, while \alg{HK-relax} and \alg{Push} maintain a superior runtime performance in local graph clustering experiments, our work constitutes an improvement in running times for flow-based methods, which in some cases can provide the best community detection results.

Before moving on we make one important comment distinguishing the implementations of \alg{SimpleLocal} and \alg{FlowSeed}. In theory both algorithms, at each step, try to find whether there exists some $S$ with $\phi_{R,\varepsilon}(S) < \alpha$ (or in the case of \alg{FlowSeed}: $\pi_R(S) < \alpha$) for some $\alpha \in (0,1)$. If they succeed, they update $\alpha \leftarrow \phi_{R,\varepsilon}(S)$ (respectively: $\alpha \leftarrow \pi_R(S)$) and repeat the process with a new $\alpha$ (see Algorithm~\ref{alg:wrapper}). However, in practice, the implementation of \alg{SimpleLocal} in the LocalGraphClustering package updates $\alpha$ by computing the standard conductance: $\alpha \leftarrow \phi(S)$. This has the advantages that it sometimes leads to fewer iterations, and guarantees that the final output set will have a \emph{standard} conductance score less than or equal to the minimum \emph{local} conductance score, though the output set may not actual minimize \emph{local} conductance~\eqref{local-cond}. In order to accurately compare the two algorithms, for our runtime experiment we also use the update $\alpha \leftarrow \phi(S)$ in our implementation of \alg{FlowSeed}. However, in all other experiments we do not make this change, since one of the key features of \alg{FlowSeed} is that it looks for sets that not only have low conductance, but also agree as much as possible with the semi-supervised information provided.

\subsection{3D Image Segmentation on a Brain Scan}

Next we turn to detecting target regions in a large graph constructed from a brain MRI. The data is made up of a labeled $256\times 287 \times 256$ MRI obtained from the MICCAI 2012 challenge~\cite{Marcus-2007-oasis}. In previous work~\cite{VeldtGleichMahoney2016} we demonstrated how to convert the image into a nearest neighbors graph on the 3D voxels. Specifically, the MRI has $256 \times 287 \times 256 \approx 18$ million voxels, with each voxel represented by an integer between 0 and 4010. For each voxel we considered its 26 spatially adjacent neighbors, i.e.\ voxels whose indices differed by at most 1 in each of the three spatial dimensions. For adjacent voxels $u$ and $v$, we computed similarities between scan intensities $I_u$ and $I_v$ using the function $e^{-(\sqrt{I_u} - \sqrt{I_v})^2/(0.05)^2}$, similar to the approach of Shi and Malik~\cite{ShiMalik2000}. These similarities were then thresholded at 0.1 and multiplied by 10 to produce a set of weighted edges in the graph with minimum weight 1. In our experiments we therefore perform calculations in terms of weighted degrees and volumes in the graph. The resulting graph has 18 million nodes and around 234 million undirected weighted edges.

The data provided by MICCAI 2012 came with 95 manually labeled regions of the brain (e.g. ventricles, amygdalas, brain stem, hippocampi, lobules, etc.). Each of these maps to a ground truth region of the brain graph. These regions range from 3104 to over 250,000 nodes in size, and with (weighted) conductance scores between 0.04 and 0.25. In our past research we showed results for detecting a single target ventricle of low conductance with around 4000 nodes~\cite{VeldtGleichMahoney2016}. Here we run semi-supervised clustering experiments on all 95 regions. More specifically, we select a set of 17 example regions out of the 95 identified regions of the brain, spanning the full range of sizes. We run extensive experiments using both the \alg{Push} algorithm~\cite{AndersenChungLang2006}, as well as \alg{FlowSeed}. We use results from experiments on the 17 regions to observe the behavior of penalizing the exclusion of seed nodes, and to guide our choice of locality parameters to choose in experiments on the remaining 78 evaluation regions.

\paragraph{Benefit of Seed Exclusion Penalties.} 
\begin{figure}[t]
	\subfloat[F1 score]
	{\includegraphics[width=.5\linewidth]{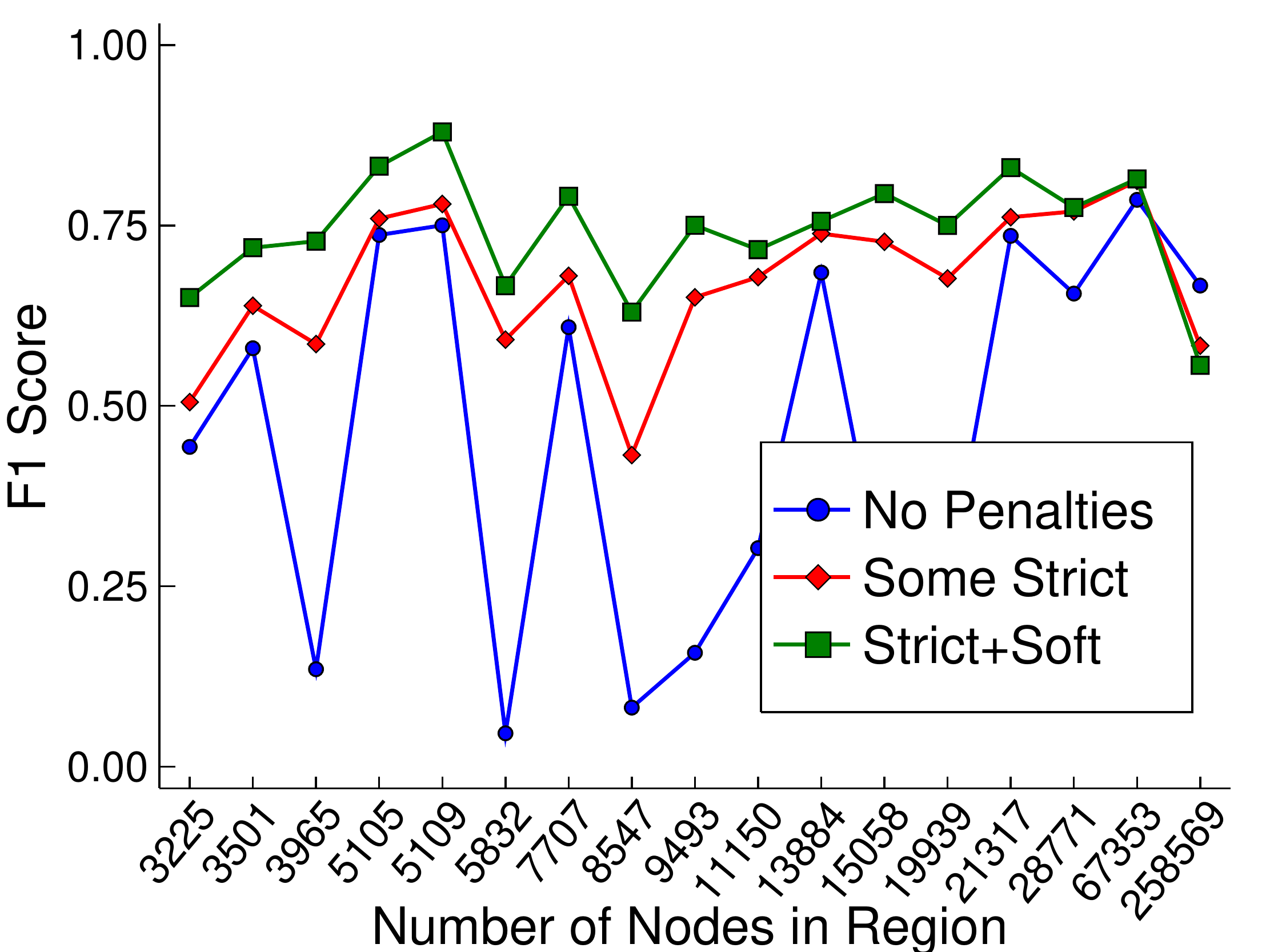}\label{bter}}\hfill
	\subfloat[Recall]
	{\includegraphics[width=.5\linewidth]{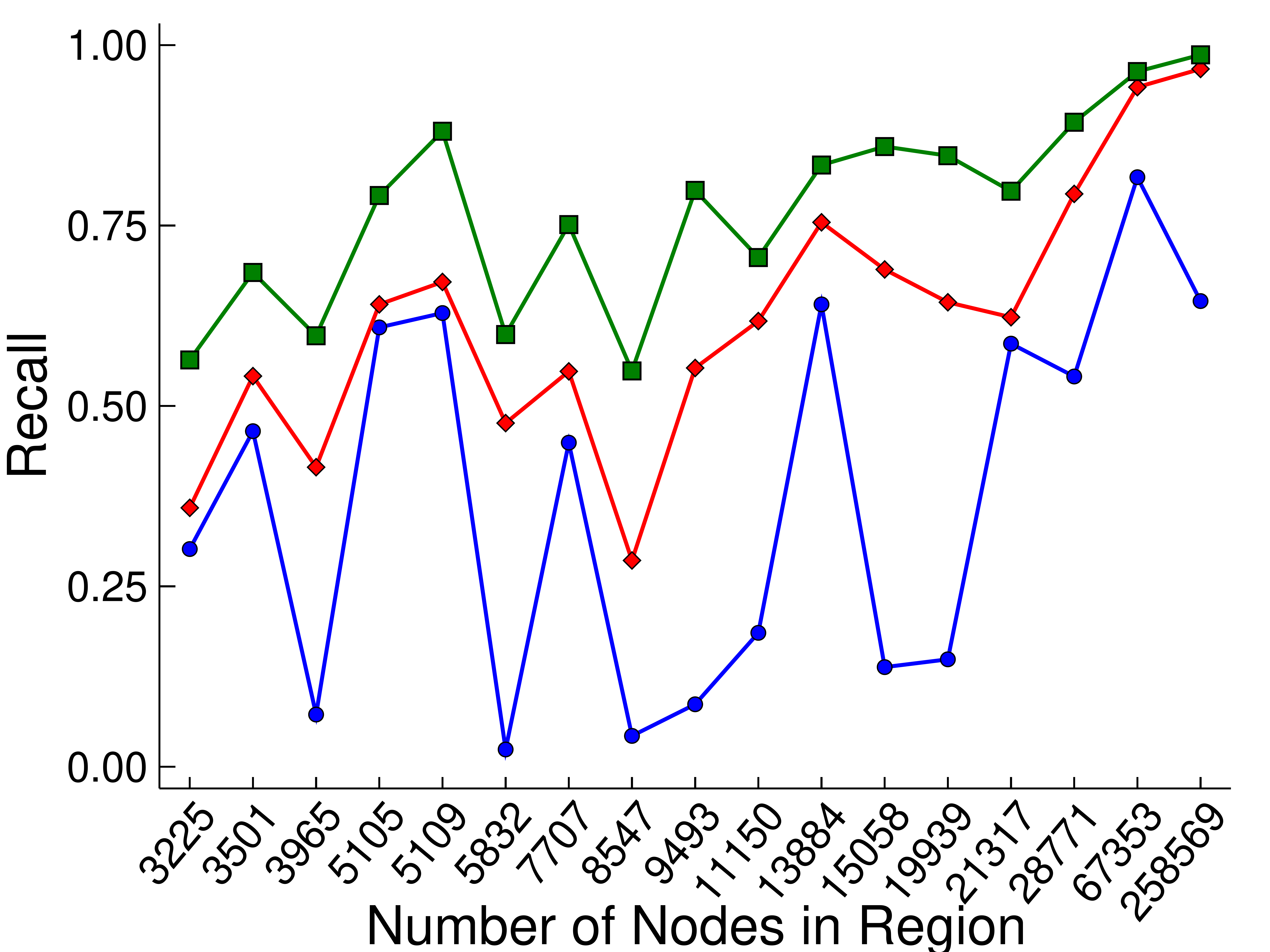}\label{cagrqcb}}
	\hfill
	\vspace{-.5\baselineskip}
	\caption{When detecting target regions of a large brain graph, enforcing no penalties on excluding seed nodes (blue) allows the method to discard too many seeds, often leading to very poor recall. If we enforce strict penalties on excluding known target nodes (red), this significantly improves recall and hence overall detection of the target cluster in terms of F1 score. We see even greater improvement by additionally including a soft penalty of $p_r = 1$ for excluding any other node in the seed set (green).}
	\vspace{-.5\baselineskip}
	\label{fig:penalties}
\end{figure}
We test a range of parameters on the 17 example regions to observe how different locality parameters and seed exclusion penalties behave for different sized regions and seed sets. We test locality parameters $\varepsilon \in \{0.5, 0.25, 0.1, 0.05\}$ and construct seed sets by taking very small subsets of the target cluster and growing them by their neighborhood. We find that with almost no exceptions, including strict and soft seed exclusion penalties leads to significant benefits in ground truth recovery across all region sizes. In Figure~\ref{fig:penalties}, for the 17 example regions we plot the recall and F1 scores for region recovery for a locality parameter of $\varepsilon = 0.1$ in the case where the seed set is made up of a random sample of $2\%$ of the target region, plus the immediate neighbors of these nodes. We run our method with (1) no penalties on excluding seed nodes, (2) strict penalties on excluding the initial $2\%$ of nodes, and (3) strict penalties on the $2\%$ and additionally a soft penalty of $p_r = 1$ for excluding any other seed nodes. As expected, when we include no penalties, the flow-based approach often shrinks the seed set into a small cluster with good conductance and very good precision, but almost no recall. As we increase the strength of seed exclusion penalties, the precision decreases slightly but the recall improves considerably, leading to a much better overall ground truth recovery.

We confirmed in numerous experiments that the same behavior holds for different locality parameters and seed sizes. For each region we form four types of seed sets. For the first we select a random set of 100 nodes from the target region, and for the remaining three we select $1\%$, $2\%$ and $3\%$ of the target region. In all cases, we grow these nodes by a one-hop neighborhood. In Figure~\ref{fig:eps1} we show F1 scores achieved by \alg{FlowSeed} on all four types of seed sets when $\varepsilon = 0.1$. 
\begin{figure}[t!]
	\subfloat[100 Target Nodes]
	{\includegraphics[width=.5\linewidth]{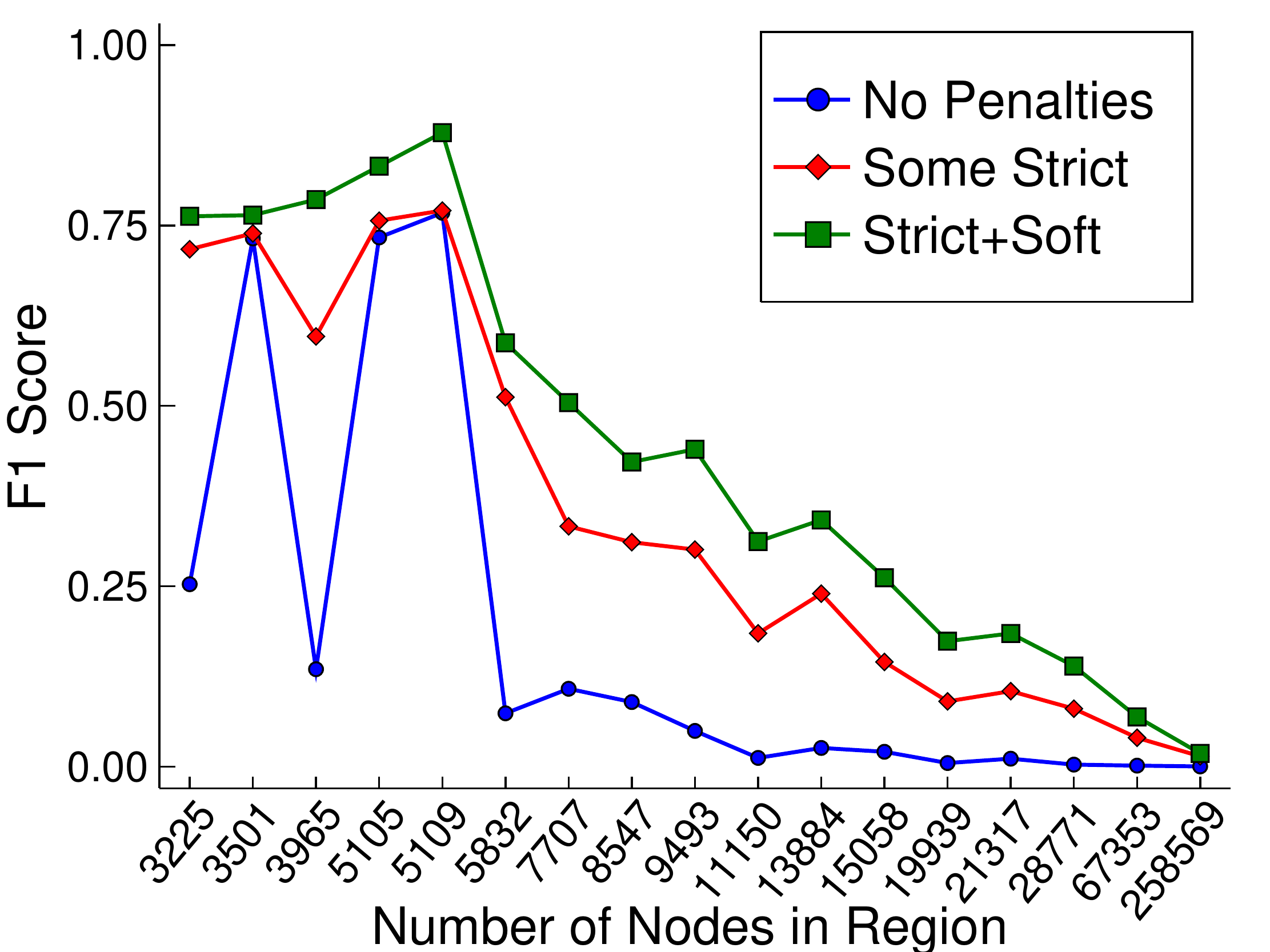}}\hfill
	\subfloat[1\% of Target]
	{\includegraphics[width=.5\linewidth]{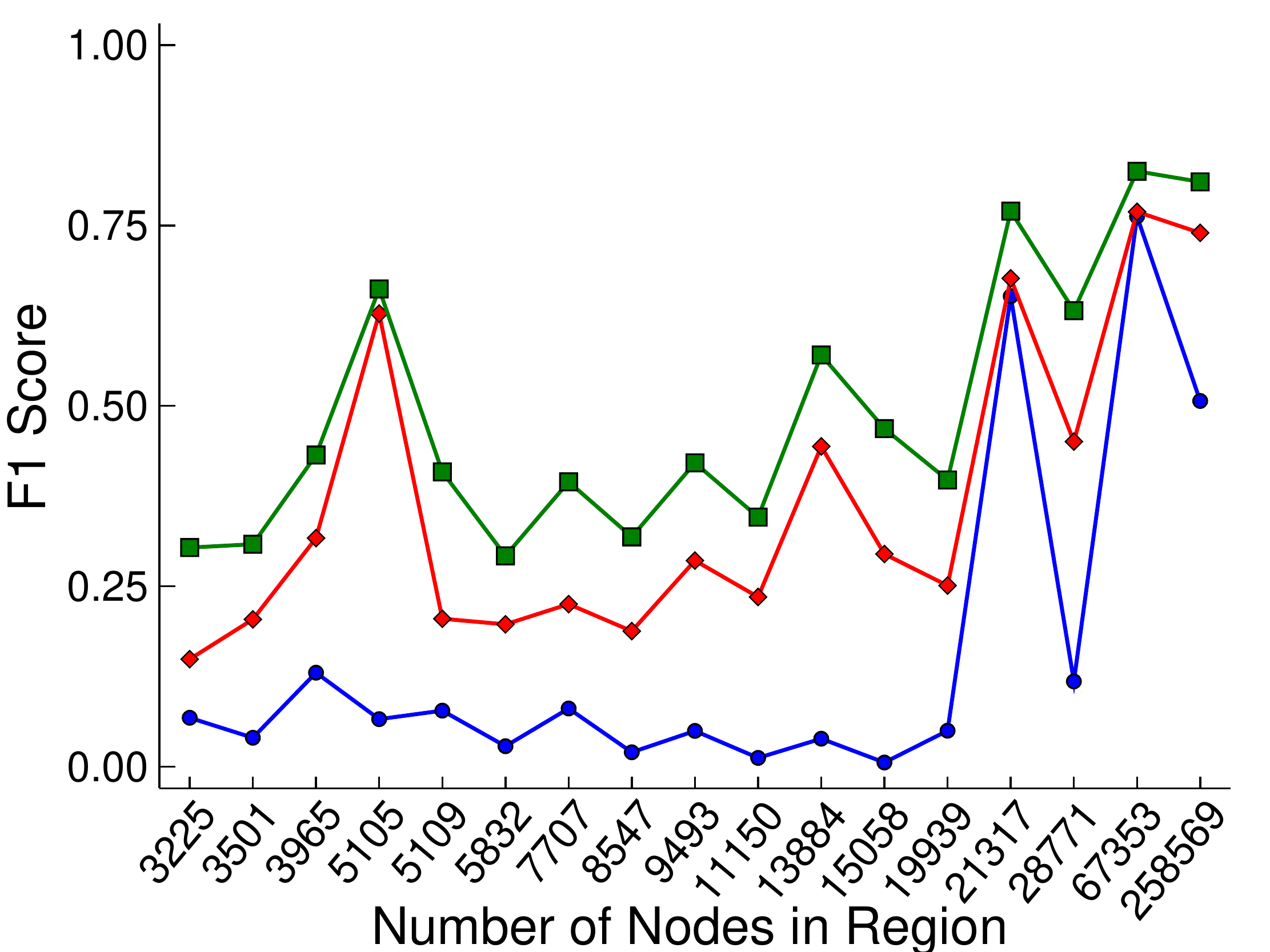}}\hfill
	\subfloat[2\% of Target]
	{\includegraphics[width=.5\linewidth]{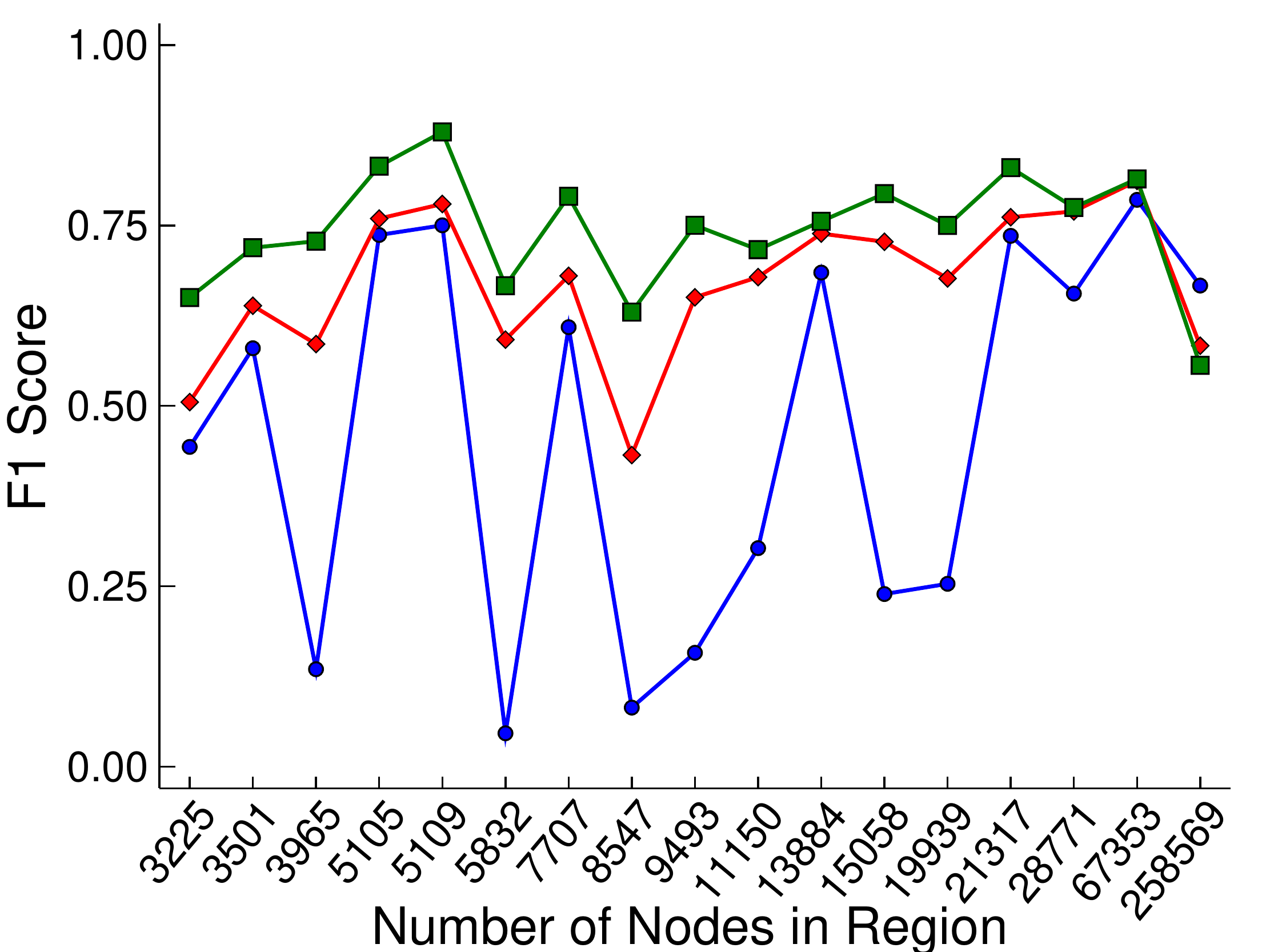}}\hfill
	\subfloat[3\% of Target]
	{\includegraphics[width=.5\linewidth]{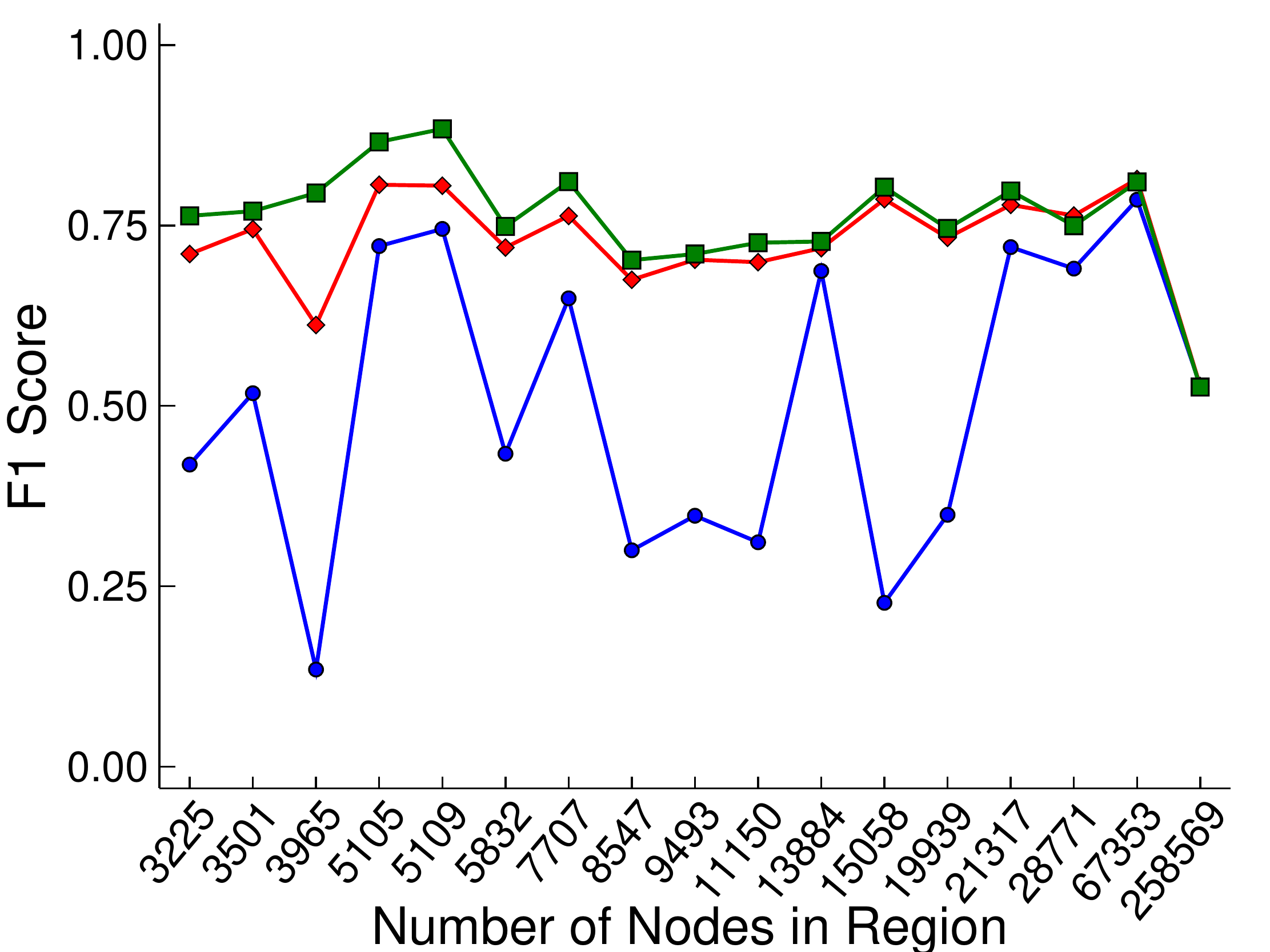}}\hfill
	\caption{F1 scores for \alg{FlowSeed} when $\varepsilon = 0.1$ on 17 brain regions using four different types of seed sets. 
		Green indicates both soft and strict penalties, the red curve shows results for just including strict penalties for nodes that are known to be in the target cluster, and the blue curve shows results for including no penalties. }
	\label{fig:eps1}
\end{figure}

\begin{figure}[t!]
	\subfloat[100 Target Nodes]
	{\includegraphics[width=.45\linewidth]{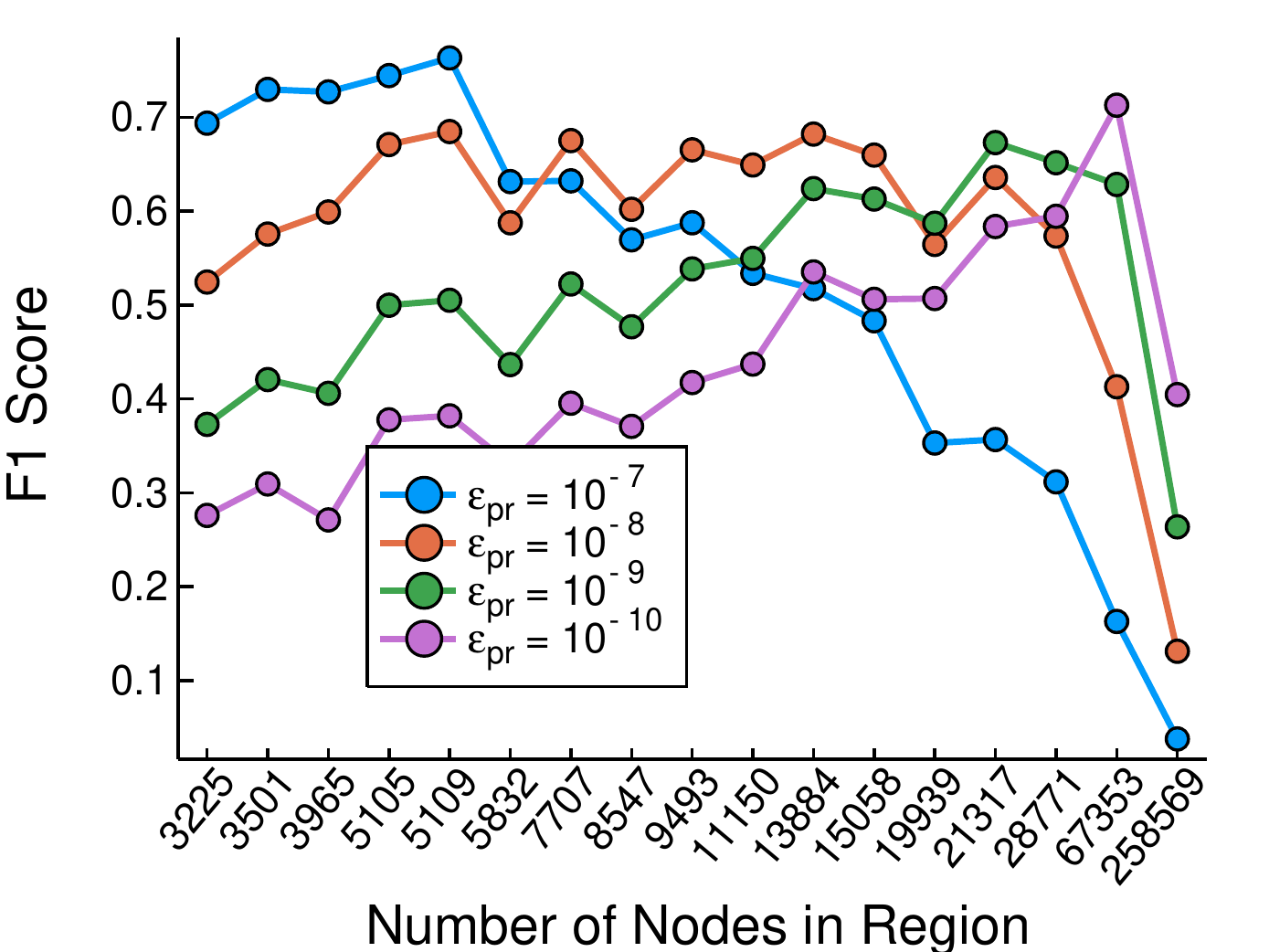}}\hfill
	\subfloat[1\% of Target]
	{\includegraphics[width=.45\linewidth]{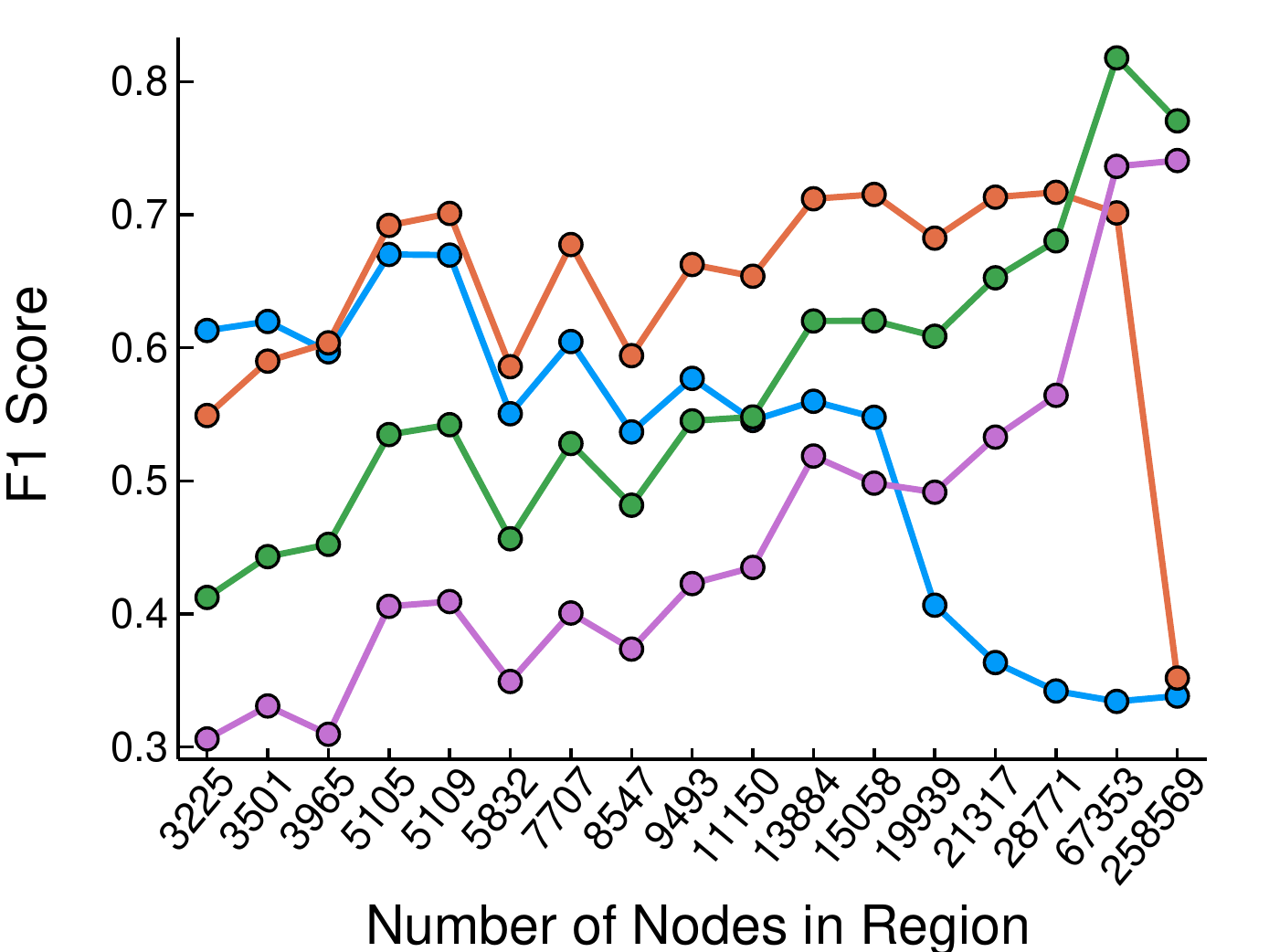}}\hfill
	\subfloat[2\% of Target]
	{\includegraphics[width=.45\linewidth]{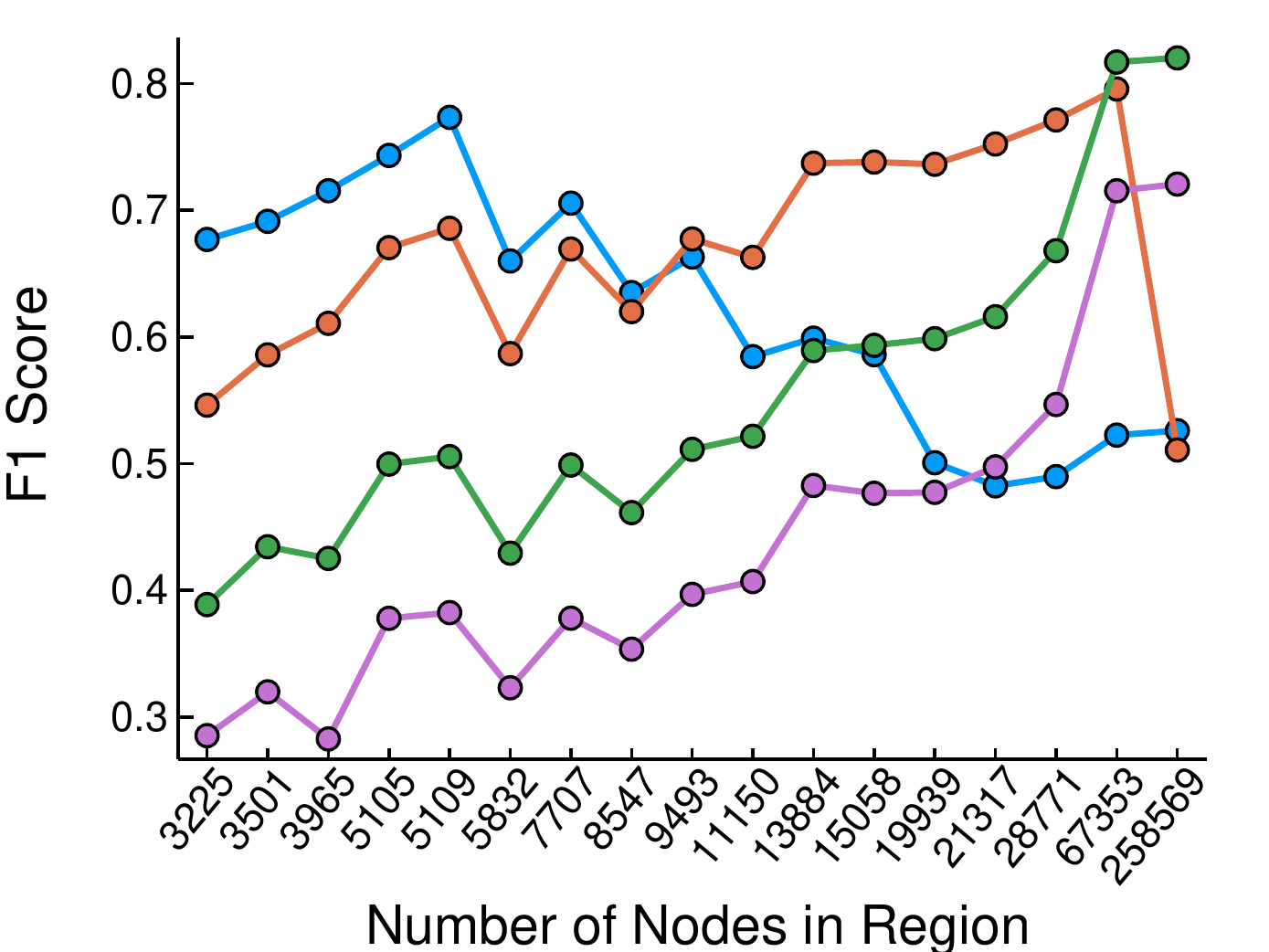}}\hfill
	\subfloat[3\% of Target]
	{\includegraphics[width=.45\linewidth]{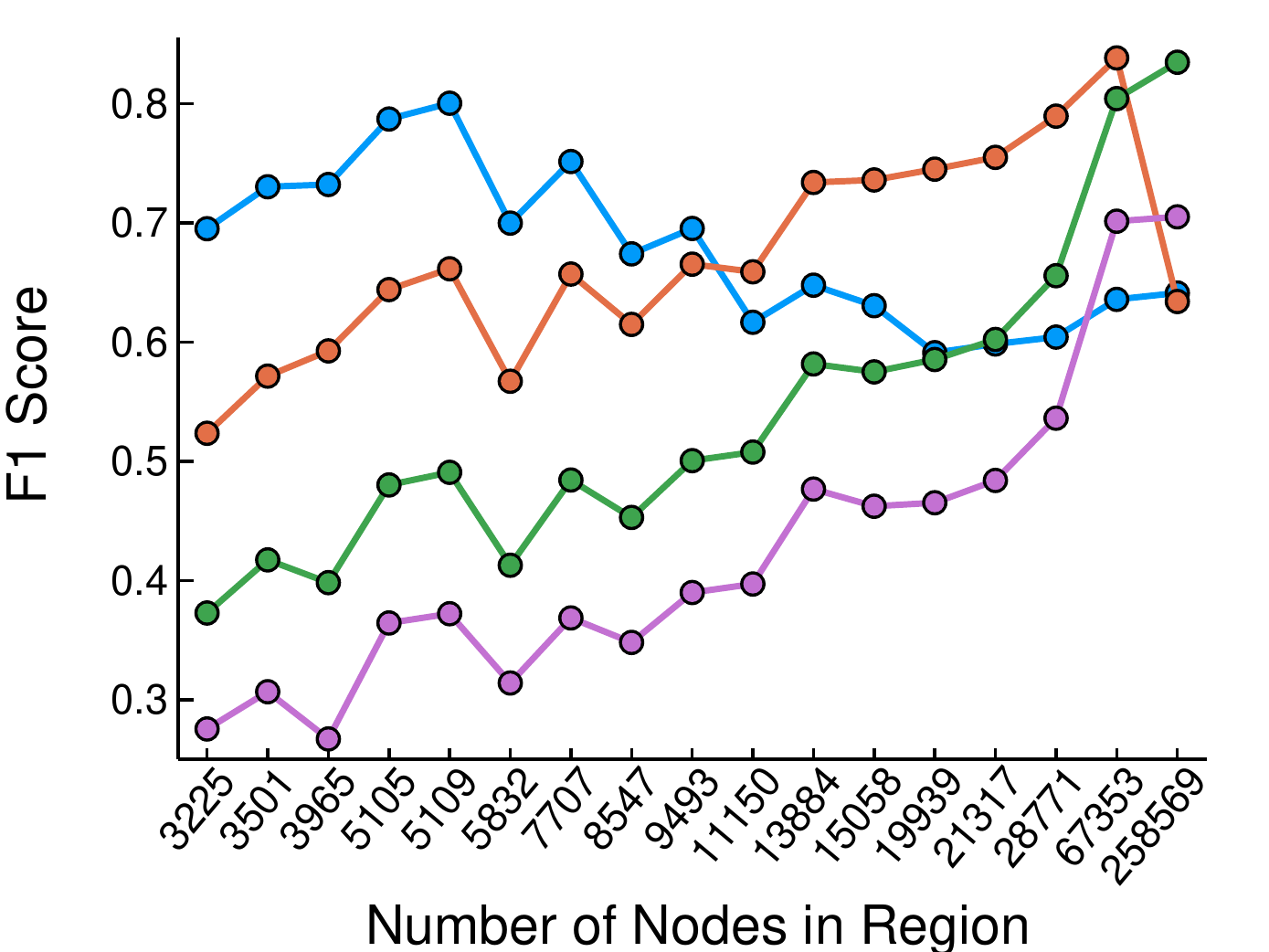}}\hfill
	\caption{F1 scores on 17 example regions of the brain graph using the PageRank \alg{Push} method with teleportation parameter $\alpha_{pr} = 0.6$ and a range of tolerance parameters $\varepsilon_{pr}$.
		Seed sets are 100 random nodes from the target region plus neighbors, or 1\%, 2\%, or 3\% of the region plus neighbors. As the region size increases, it becomes necessary to use smaller values of $\varepsilon_{pr}$ to accurately identify the target cluster.}
	\label{fig:ppr}
\end{figure}

\paragraph{Comparison with Random-Walk Methods} We additionally run the PageRank \alg{push} algorithm~\cite{AndersenChungLang2006} with teleportation parameters $\alpha_{pr}$ from $0.5$ to $0.9$, and approximate PageRank tolerance parameters $\varepsilon_{pr}$ from $10^{-11}$ to $10^{-7}$. In practice we find that smaller values of $\alpha_{pr}$ perform better, with little difference between parameters between 0.5 and 0.7. The values of $\varepsilon_{pr}$ we use here are significantly smaller than the ones we used for experiments in Section~\ref{cd}. This is because the MRI graph is much more structured and geometric than the real-world networks we considered in Section~\ref{cd}, and thus there are large sets of nodes with good topological community structure we wish to find in the MRI graph. To find these, we need to explore a wider region of the graph, and hence we must use small tolerance parameters.
We show results for all seed sizes for $\alpha_{pr} = 0.6$ in Figure~\ref{fig:ppr}.

For both \alg{Push} and \alg{FlowSeed}, we use observations from the experiments on the 17 example regions to inform our choice of parameter settings for different sized regions and seed set sizes. We then use these parameters to test the performance of each method on the remaining 78 regions, which we refer to as the evaluation set. We run experiments for the case where we know exactly 100 of the target nodes, and where $1\%, 2\%,$ and $3\%$ of the target region is given. We run \alg{Push} with a teleportation parameter of $\alpha_{pr} = 0.6$, and run \alg{FlowSeed} with both strict and soft penalties. In each experiment we identify which of the 17 example regions is closest in size to the target region from the evaluation set, and then set $\varepsilon$ and $\varepsilon_{pr}$ to be the values which led to the best F1-score recovery for this comparable example region. We plot results for all types of seed sets on a subset of the 78 regions (for easier display) in Figure~\ref{fig:test}.

\begin{figure}[t!]
	\centering
	\subfloat[100 Target Nodes]
	{\includegraphics[width=.45\linewidth]{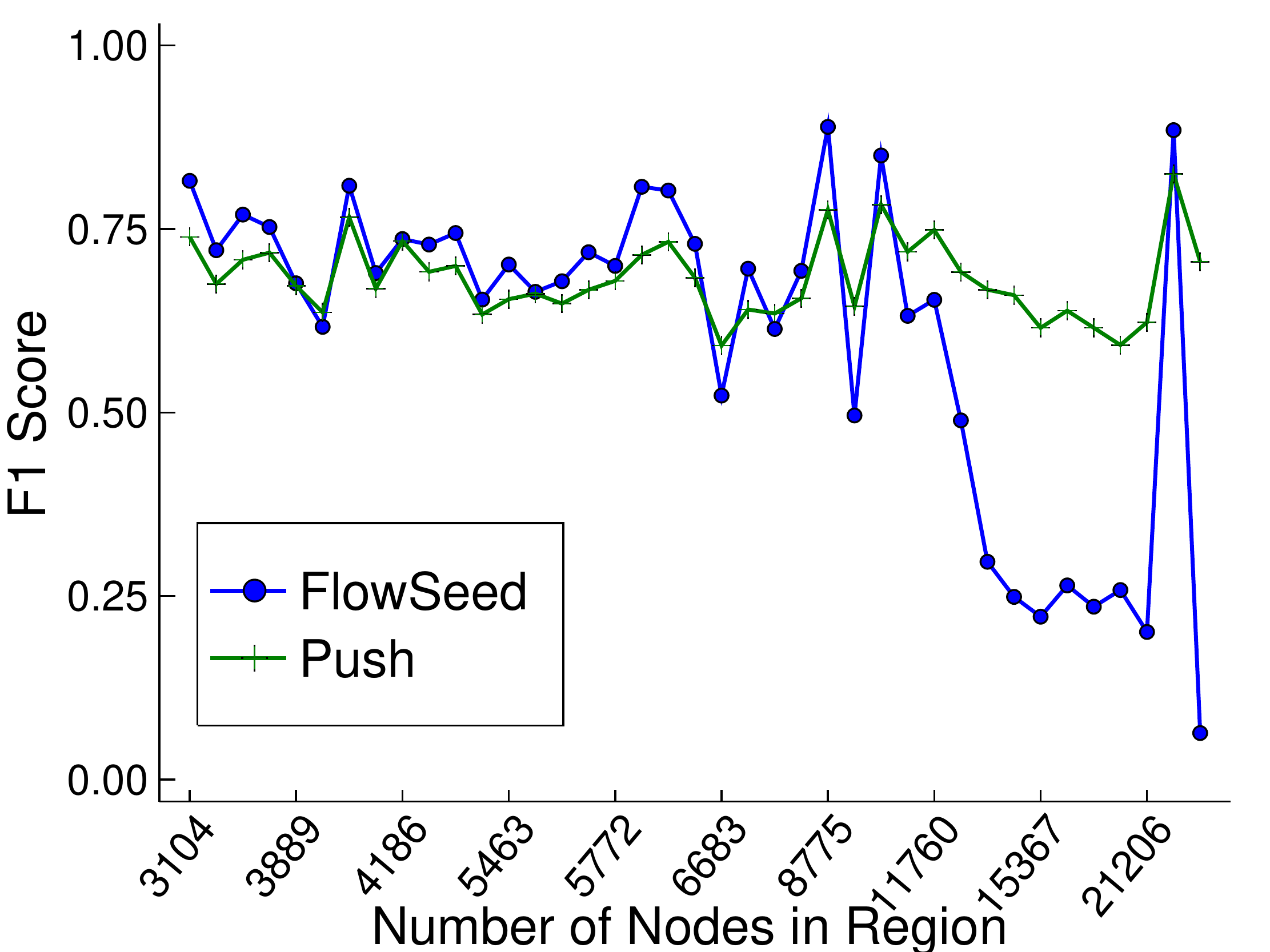}}\hfill
	\subfloat[1\% of Target]
	{\includegraphics[width=.45\linewidth]{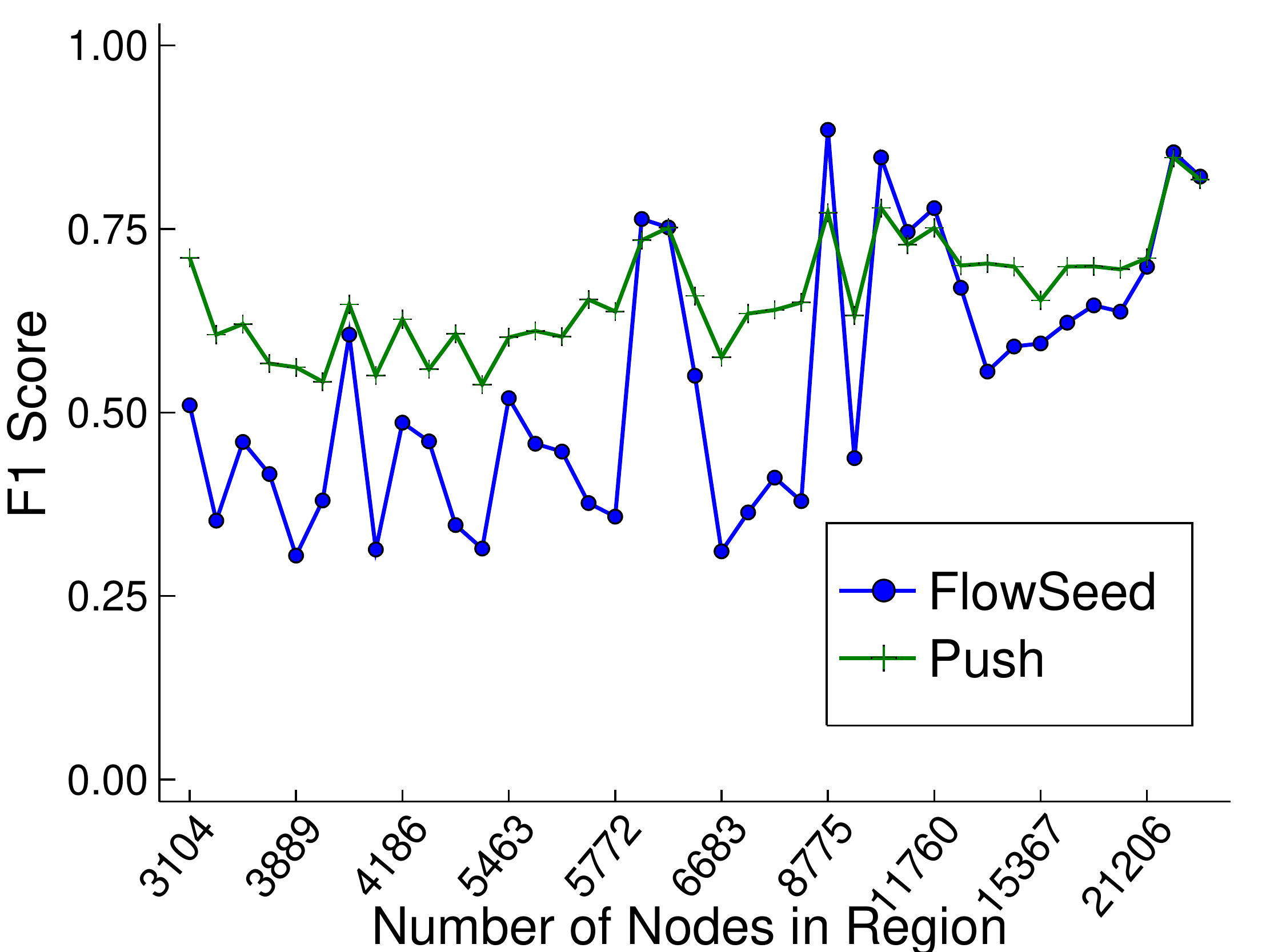}}\hfill
	\subfloat[2\% of Target]
	{\includegraphics[width=.45\linewidth]{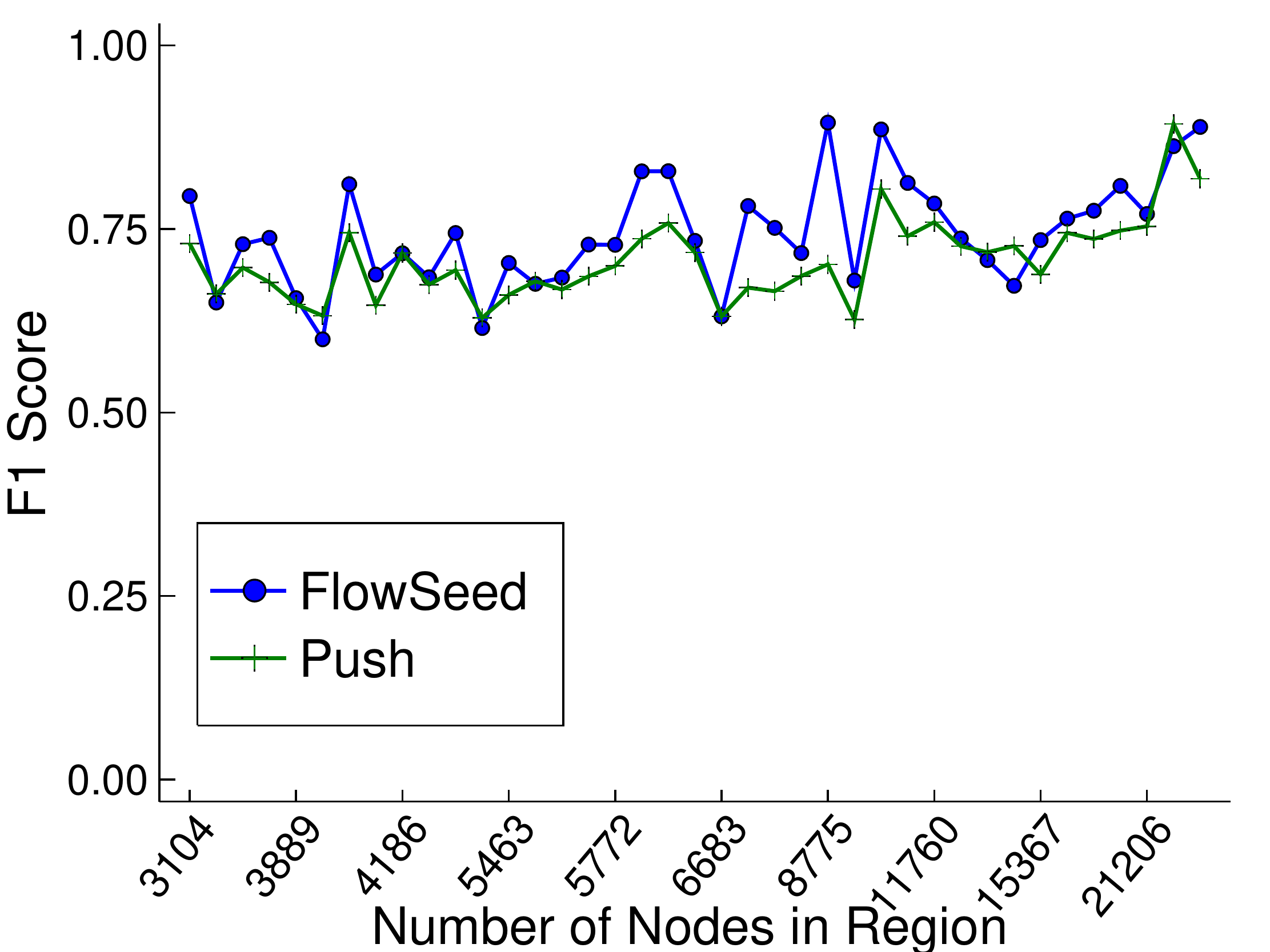}}\hfill
	\subfloat[3\% of Target]
	{\includegraphics[width=.45\linewidth]{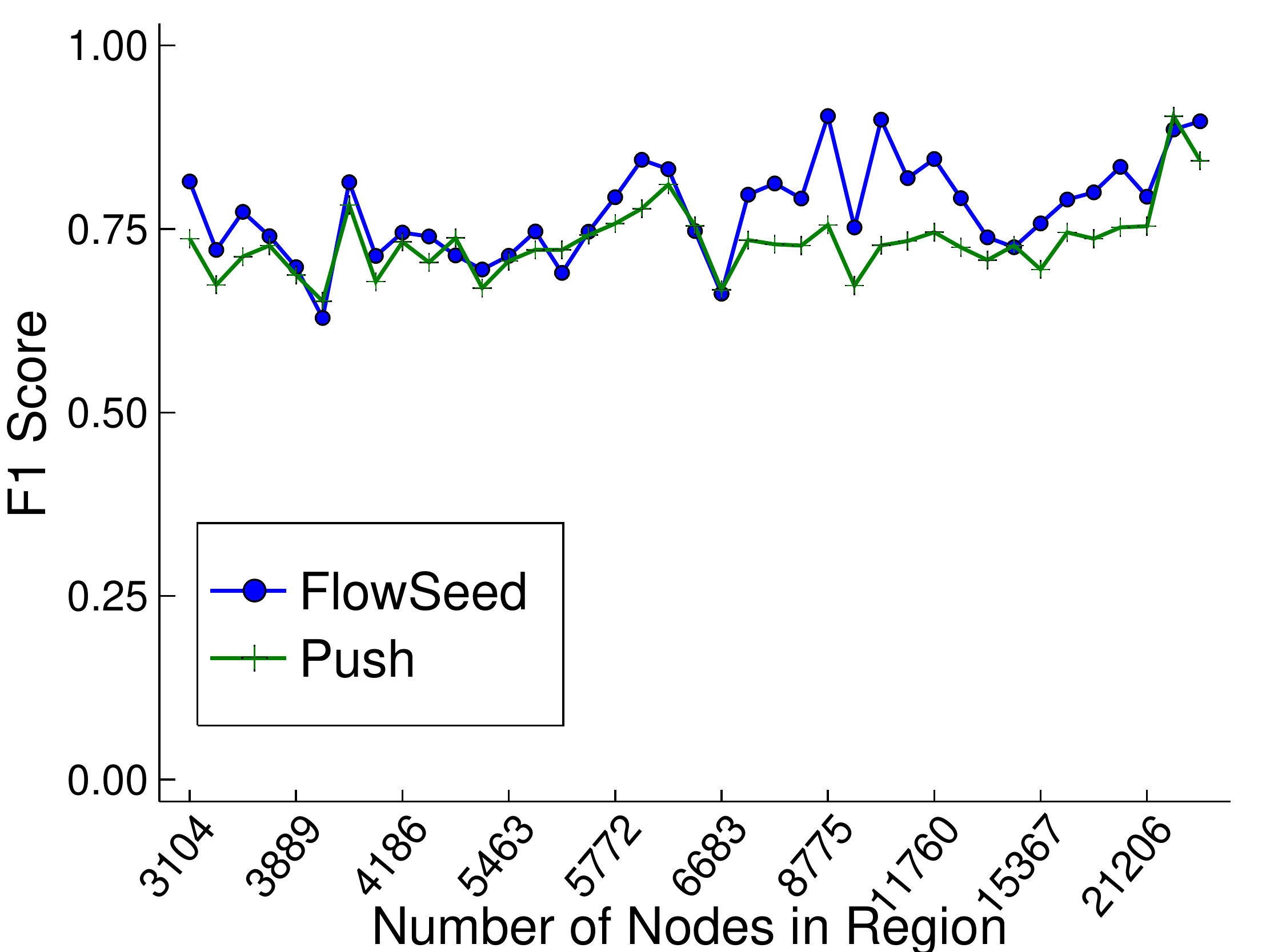}}\hfill
	\caption{F1 scores for both \alg{FlowSeed} and \alg{Push} on a subset of half the testing regions on the brain graph. Most region sizes are omitted from the $x$-axis for easier display. When enough of the target nodes are known (e.g. (c) and (d) and the first part of (a)), then \alg{FlowSeed} is able to outperform \alg{Push} in identifying target regions. When only a small amount of the target set is known, \alg{Push} performs better as it is able to quickly grow the seed set into a large enough set to capture most of the target region.
	}
	\label{fig:test}
\end{figure}

Our experiments highlight a tradeoff in the performance of the two algorithms. For small seed sets (100 target nodes plus their neighborhood, or 1\% of the target region plus neighbors), \alg{Push} typically outperforms \alg{FlowSeed} in ground truth recovery, as it is able to grow a very small seed set into a sizable cluster. However, given sufficient information regarding the target cluster, we see a distinct benefit in applying our flow-based approach. When 2\% (resp. 3\%) of the target is known, \alg{FlowSeed} obtains a higher F1 score for 64 (resp. 68) of the 78 target regions, and the scores are on average 6.2\% (resp. 6.1\%) higher than those returned by \alg{Push}.
In terms of runtime, the highly optimized \alg{Push} implementation is faster: most experiments run in under 1 second, with the largest taking several seconds. Our method takes up to 15 minutes for the largest region, but typically runs in 10-60 seconds for small and medium sized regions.


\subsection{Detecting an Atrial Cavity}

In our last experiment we demonstrate that random-walk and flow-based methods can be viewed as complementary approaches rather than competing algorithms. We combine the strengths of \alg{Push} and \alg{FlowSeed} to provide good quality 3D segmentations of a manually labeled left atrial cavity in a whole-body MRI scan. The dataset was provided as a part of the 2018 Atrial Segmentation Challenge, which sought efficient methods for automatic segmentation of the atrial cavity for clinical usage~\cite{xiong2018fully}. We convert one such MRI into a graph with 29.2 million nodes and 390 million edges using the same technique used to construct the brain graph. The cavity in the MRI corresponds to a target cluster with 252,364 nodes and a conductance of 0.0414 in the graph.

We begin from a small set of 100 randomly selected nodes from the atrial cavity, constituting less than 0.04\% of the target region. We grow these nodes by a one-hop and a two-hop neighborhood to produce two different seed sets to use as input to \alg{Push}. The algorithm's performance is very similar using both seed sets, so we just report results using the two-hop neighborhood. We again set $\alpha_{pr} = 0.6$ and test a range of tolerance parameters $\varepsilon_{pr}$ from $10^{-14}$ to $10^{-8}$. Looking at results in Figure~\ref{fig:cavity_push}, we see that the \alg{Push} algorithm is simply growing circular regions around seed nodes. Many of the output sets are not connected. From Table~\ref{tab:cavity}, we note that the best F1 score achieved by \alg{Push} is 0.714 when $\varepsilon_{pr} = 10^{-12}$.

\begin{figure}[t]
	\subfloat[\alg{Push} $\varepsilon_{pr} = 10^{-9}$]
	{\includegraphics[width=.25\linewidth]{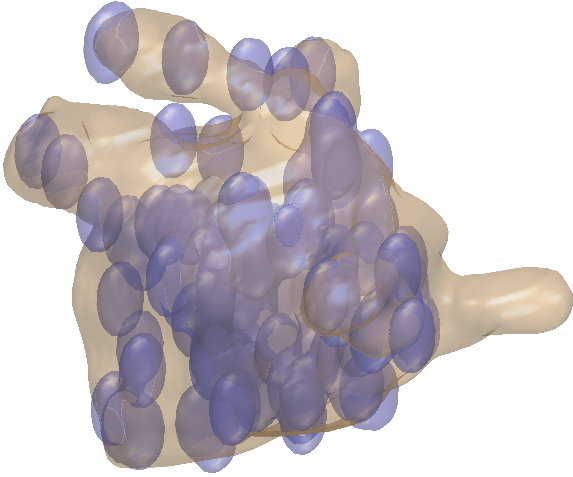}\label{del9}}\hfill
	\subfloat[\alg{Push} $\varepsilon_{pr} = 10^{-10}$]
	{\includegraphics[width=.25\linewidth]{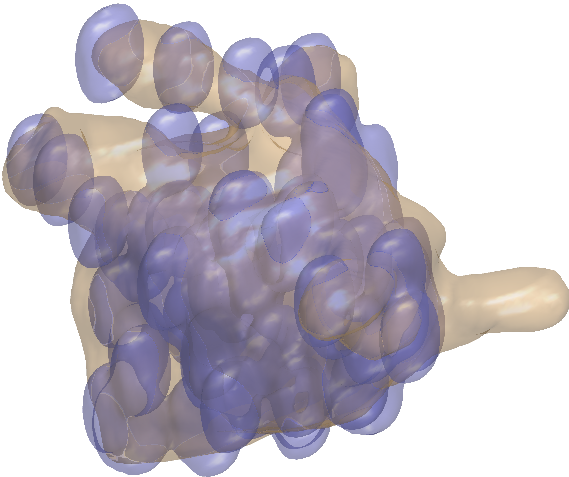}\label{del10}}\hfill
	\subfloat[\alg{Push} $\varepsilon_{pr} = 10^{-11}$]
	{\includegraphics[width=.25\linewidth]{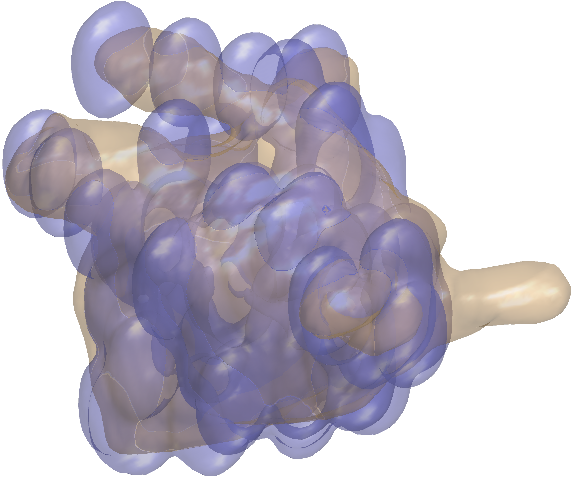}\label{del11}}\hfill
	\subfloat[\alg{Push} $\varepsilon_{pr} = 10^{-12}$]
	{\includegraphics[width=.25\linewidth]{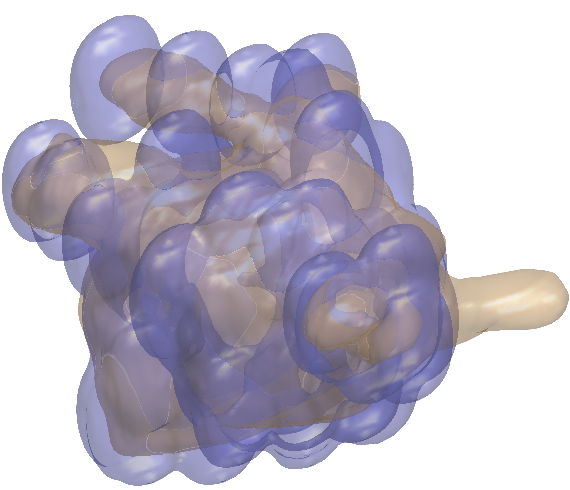}\label{del12}}\hfill
	\caption{The yellow region indicates the target cavity in a 29 million node graph constructed from a full-body MRI scan. Purple regions indicate sets returned by the \alg{Push} algorithm. Starting from a 100 random nodes in the target set plus their neighborhood, \alg{Push} grows circular regions which grow as $\varepsilon_{pr}$ decreases. }
	\label{fig:cavity_push}
\end{figure}

\begin{figure}[t]
	\subfloat[{\alg{FlowSeed} refinement of \alg{Push} ($\varepsilon_{pr} = 10^{-9}$)}]
	{\includegraphics[width=.275\linewidth]{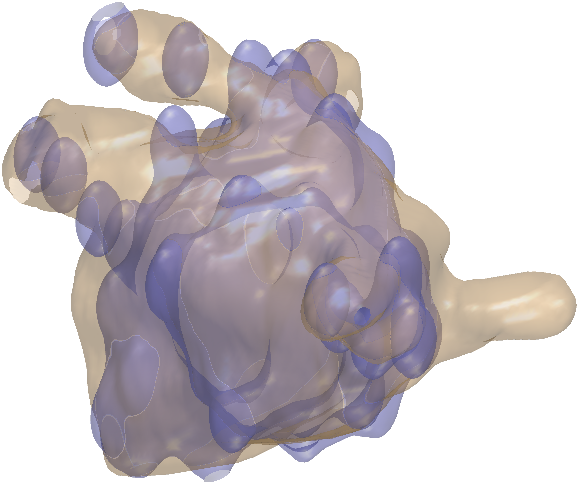}\label{sep8}}\hfill
	\subfloat[\alg{FlowSeed} refinement of \alg{Push} ($\varepsilon_{pr} = 10^{-10}$)]
	{\includegraphics[width=.275\linewidth]{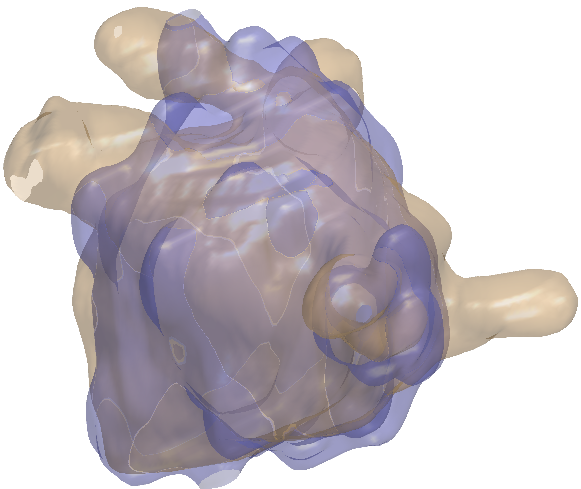}\label{sep9}}\hfill
	\subfloat[\alg{FlowSeed} refinement of \alg{Push} ($\varepsilon_{pr} = 10^{-11}$)]
	{\includegraphics[width=.275\linewidth]{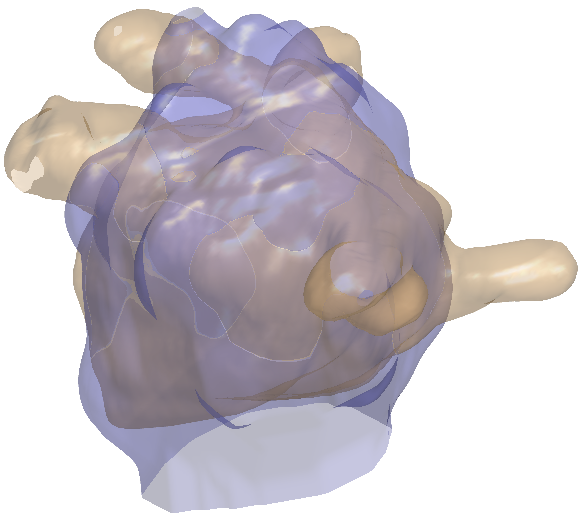}\label{sep10}}\hfill
	\caption{On the atrial cavity dataset, we refine the output of \alg{Push} (see Figure~\ref{fig:cavity_push}) using \alg{FlowSeed} with a locality parameter $\varepsilon = 0.1$. \alg{FlowSeed} (purple) fills in the interior of the target region (yellow) and more closely identifies the region boundary.}
	\label{fig:cavity_fs}
\end{figure}

\begin{table}[t!]
	\caption{Results for detecting a target atrial cavity in a graph constructed from a full-body MRI. Letting \alg{Push} expand a small seed set and then refining the output with \alg{FlowSeed} leads to better F1 scores than simply running \alg{Push} with smaller $\varepsilon_{pr}$.}
	\label{tab:cavity}
	\centering
	\begin{tabular}{lrrrrrr}
		\toprule
		method & $\varepsilon_{pr}$ & size & pr & re & F1 & time \\
		\midrule
		\alg{Push} & $10^{-9}$ & 72337  & 0.849& 0.243 & 0.378 & 1\\ 
		+\alg{FlowSeed} &  ($\varepsilon = .1$)& 160951  &  0.924 & 0.590 & 0.720  &1410  \\ 
		\midrule
		\alg{Push} & $10^{-10}$ & 133618 & 0.792 & 0.419 & 0.549 & 3 \\ 
		+\alg{FlowSeed} & ($\varepsilon = .1$) & 224842 &  0.850&	0.757 &0.801 & 3573\\ 
		\midrule
		\alg{Push} & $10^{-11}$ & 211571 & 0.732 & 0.614 & 0.668 & 4 \\ 
		+\alg{FlowSeed} & ($\varepsilon = .1$) & 296192 &  0.690 &	0.809	& 0.745 & 9800\\ 
		\midrule
		\alg{Push} & $10^{-12}$ & 290937 & 0.666 & 0.768 & 0.714 & 5 \\ 
		\alg{Push} & $10^{-13}$ & 367011 & 0.599 & 0.871 & 0.710 & 6 \\ 
		\bottomrule
	\end{tabular}
\end{table}

We next take the output of \alg{Push} and refine it using \alg{FlowSeed} with locality parameter $\varepsilon = 0.1$. We set a strict penalty on excluding the original 100 nodes from the cavity, a soft penalty of 1 on excluding their neighbors, and a penalty of 0.5 on excluding any node in the set returned by \alg{Push}. We see a significant improvement in quality of segmentation, in the best case leading to a precision of 0.8498, recall of 0.7571, and F1 score of 0.8008 when refining the region output by \alg{Push} when $\varepsilon_{pr} = 10^{-10}$ (see Table~\ref{tab:cavity}). In Figure~\ref{fig:cavity_fs} we see that \alg{FlowSeed} smooths out the circular regions returned by \alg{Push} to return a connected region that better identifies the boundary of the target cavity. Regarding runtime, \alg{Push} quickly grows circular regions within a few seconds, and the \alg{FlowSeed} refinement procedure takes just under an hour. Together these methods produce a significantly better output than either method could have accomplished alone.

\section{Conclusions and Future Work}
Flow-based algorithms for local graph clustering exhibit very strong cut improvement and runtime guarantees. In our work we have exploited efficient warm-start and push-relabel heuristics to provide practitioners with a very simple yet fast flow-based method for real-world data mining applications. In addition to outperforming related flow-based clustering algorithms in runtime, our method is able to better incorporate domain-specific semi-supervised information about ground truth target clusters in a large network, by giving users the option to specify penalties and strict constraints for excluding specific seed nodes from the output set. Given the success of seed exclusion penalties for flow-based methods, in future work we will continue to explore how similar penalties may be incorporated in other well-known clustering approaches including spectral and random-walk based techniques.

\bibliography{localclustering}
\bibliographystyle{plain}
\end{document}